\documentclass[submission,copyright,creativecommons]{eptcs}
\providecommand{\event}{GandALF 2021} 

\title{Stochastic Games with Disjunctions of Multiple~Objectives\thanks{This work was supported by the DFG RTG 2236 ``UnRAVeL'' (Winkler) and DFG projects 383882557 ``SUV'' and 427755713 ``GOPro'' (Weininger).}}
\author{Tobias Winkler
	\institute{RWTH Aachen University, Germany}
	\email{tobias.winkler@cs.rwth-aachen.de}
	\and
	Maximilian Weininger
	\institute{Technical University of Munich, Germany}
	\email{maxi.weininger@tum.de}
}

\def\titlerunning{Stochastic Games with Disjunctions of Multiple~Objectives}
\def\authorrunning{T. Winkler and M. Weininger}

\usepackage{amsmath,amssymb,xcolor,graphicx,wrapfig}
\usepackage{algorithmicx}
\usepackage[noend]{algpseudocode}
\usepackage{booktabs}
\usepackage{proof}
\usepackage{tikz}
\usepackage{csquotes}
\usepackage{graphics}
\usepackage{stmaryrd}
\usepackage{subcaption}
\usepackage{url}
\usepackage{bbm}

\usepackage{pgfplots}

\usepackage{tabu,multirow}
\usepackage{xparse}
\usepackage{mathtools}
\usepackage{environ}

\usepackage{caption}
\usepackage{subcaption}

\usepackage{algorithm}
\usepackage{algpseudocode}

\usepackage{marginfix}
\usepackage{xifthen}

\usepackage{xfrac}

\usepackage{amsthm}

\usepackage{nicefrac}

\usepackage{thmtools} 
\usepackage{thm-restate}

\usetikzlibrary{shapes,positioning,arrows,calc,automata,matrix,fit,arrows.meta}
\tikzstyle{state}+=[minimum size = 6mm, inner sep=0,outer sep=1]
\tikzset{->,>=stealth'}

\usetikzlibrary{
    shapes,
    automata,
    arrows,
    arrows.meta,
    calc
}

\tikzset{
    diagonal fill/.style 2 args={fill=#2, path picture={
            \fill[#1, sharp corners] (path picture bounding box.south west) -|
            (path picture bounding box.north east) -- cycle;}},
    reversed diagonal fill/.style 2 args={fill=#2, path picture={
            \fill[#1, sharp corners] (path picture bounding box.north west) |- 
            (path picture bounding box.south east) -- cycle;}}
}

\tikzstyle{max}=[rectangle,draw=black,thick,inner sep=.8mm,minimum size=6mm]
\tikzstyle{target}=[accepting] 
\tikzstyle{min}=[diamond,draw=black,thick,inner sep=0.5mm,minimum size=8mm]
\tikzstyle{prob}=[circle,draw=black,thick,inner sep=.8mm,minimum size=6mm] 
\tikzstyle{mcstate}=[circle,draw=black,thick,inner sep=1mm]
\tikzstyle{state}=[circle,draw=black,thick,inner sep=1mm,minimum size=8mm] 
\tikzstyle{trans}=[]
\tikzstyle{col1}=[fill=yellow!35]
\tikzstyle{col2}=[fill=green!20]
\tikzstyle{col3}=[fill=blue!15]
\tikzstyle{both cols}=[diagonal fill={green!20}{yellow!35}]
\tikzstyle{both cols2}=[diagonal fill={green!20}{blue!15}]
\tikzstyle{smaller}=[scale=0.75]
\tikzset{every loop/.style={min distance=5mm,looseness=2.5}}

\tikzstyle{pcurve}=[-, thick,line cap= round]

\newcommand{\drawaxes}{
    \draw[black!55,->] (0,0) -- (0,1.4);
    \draw[black!55,->] (0,0) -- (1.4,0);
}

\DeclareGraphicsExtensions{.pdf, .png, .jpg}

\newcommand{\Naturals}{\mathbb{N}}
\newcommand{\Reals}{\mathbb{R}}

\DeclareDocumentCommand{\post}{D<>{} O{} D(){}}{\mathsf{Post}_{#1}^{#2}\ifthenelse{\isempty{#3}}{}{(#3)}}
\newcommand{\eqdef}{\vcentcolon=}
\newcommand{\defeq}{=\vcentcolon}

\newcommand{\NP}{\mathsf{NP}}
\renewcommand{\P}{\mathsf{P}}
\newcommand{\coNP}{\mathsf{coNP}}
\newcommand{\EXPTIME}{\mathsf{EXPTIME}}
\newcommand{\NEXP}{\mathsf{NEXPTIME}}
\newcommand{\PSPACE}{\mathsf{PSPACE}}

\tikzstyle{myArrowStyle} = [shorten >=1pt,->,>=stealth']

\newcommand{\conv}{conv}
\newcommand{\dwc}{\mathsf{dwc}}

\newcommand{\genobj}{\star}
\newcommand{\safe}{\square}
\newcommand{\reach}{\lozenge}
\newcommand{\straa}{\sigma}
\newcommand{\strab}{\tau}

\newcommand{\probability}{\mathbb{P}}
\newcommand{\expectation}{\mathbb{E}}

\newcommand{\val}{\mathcal{A}}

\newcommand{\Achiever}{Eve}
\newcommand{\Spoiler}{Adam}

\newcommand{\query}{\varphi}
\newcommand{\targets}{\mathcal{T}}
\newcommand{\game}{\mathcal{G}}
\newcommand{\maxstrat}{\straa}

\newcommand{\minstrat}{\strab}

\newcommand{\sinit}{{s_0}}
\newcommand{\choices}{\mathsf{Act}}
\newcommand{\Dist}{\mathsf{Dist}}
\newcommand{\supp}{\mathsf{supp}}

\renewcommand{\path}{\pi}

\newcommand{\achsymb}{\mathsf{E}}
\newcommand{\spoisymb}{\mathsf{A}}

\newcommand{\uviop}{\Phi}
\newcommand{\pcurves}{\mathfrak{C}}
\newcommand{\pset}[1]{\mathcal{P}(#1)}
\newcommand{\uvidomelem}{\mathcal{X}}
\renewcommand{\conv}{\mathsf{conv}}
\newcommand{\Smax}{S_{\achsymb}}
\newcommand{\Smin}{S_{\spoisymb}}
\newcommand{\Sprob}{S_{\mathsf{P}}}
\newcommand{\inpw}{\ \dot{\in}\ }

\newcommand{\upval}{\val^{\forall \exists}}

\newcommand{\unf}{\mathsf{Unf}}

\newtheorem{definition}{Definition}
\newtheorem{example}{Example}

\newtheorem*{claim*}{Claim}
\newenvironment{claimproof}{
	
	\textit{$\triangleright$ Proof:}}{\hfill$\triangleleft$

}

\newtheorem{lemma}{Lemma}
\newtheorem{observation}{Observation}
\newtheorem{theorem}{Theorem}
\newtheorem{corollary}{Corollary}

\newtoggle{arxiv}
\toggletrue{arxiv}

\begin{document}

\maketitle

\begin{abstract}
Stochastic games combine controllable and adversarial non-determinism with stochastic behavior and are a common tool in control, verification and synthesis of reactive systems facing uncertainty.
Multi-objective stochastic games are natural in situations where several---possibly conflicting---performance criteria like time and energy consumption are relevant.
Such \emph{conjunctive} combinations are the most studied multi-objective setting in the literature.
In this paper, we consider the dual \emph{disjunctive} problem.
More concretely, we study turn-based stochastic two-player games on graphs where the winning condition is to guarantee at least one reachability or safety objective from a given set of alternatives.
We present a fine-grained overview of strategy and computational complexity of such \emph{disjunctive queries} (DQs) and provide new lower and upper bounds for several variants of the problem, significantly extending previous works.
We also propose a novel value iteration-style algorithm for approximating the set of Pareto optimal thresholds for a given DQ.
\end{abstract}

\section{Introduction}
\label{sec:intro}

Stochastic games (SG), e.g.\ \cite{Condon90,shapley1953}, combine controllable and adversarial non-determinism with stochastic behavior.
In their turn-based two-player version, SGs are played on graphs where the vertices are called states, and every state either belongs to one of the two players \Achiever\ and \Spoiler, or is controlled by a probabilistic environment.
In each round, the player in control of the state chooses an action---an edge of the graph---and the game transitions to the successor state.
In probabilistic states, the successor is sampled according to a fixed observable distribution over the outgoing edges.
\emph{Simple} SGs~\cite{Condon90} have just a single \emph{reachability} objective.
A key question is whether \Achiever\ can control her states such that a target is reached with at least a certain probability, no matter the behavior of the opponent \Spoiler. 
Dually, \Spoiler\ has a \emph{safety objective} in this setting:
He should maximize the chance of staying in the---from his point of view---safe region, i.e., avoiding Eve's target states.

\emph{Multi-objective stochastic games}~\cite{cfk13} extend this by allowing Boolean combinations of several different probability or expectation thresholds on various objectives.
Such games have been used to synthesize optimal controllers in application scenarios where the system at hand is exposed to an environment with both stochastic and non-deterministic aspects~\cite{concur14basset,DBLP:conf/qest/ChenKSW13,FengWHT15}.
A natural subclass of multi-objectives are \emph{disjunctive queries} (DQ)~\cite{cfk13,FH10} where the player has to satisfy at least one alternative from a given set of options.

In this paper, we study DQs with both reachability and safety options.
More specifically, given a game and a finite set of reachability and safety objectives, each equipped with a desired threshold probability, we ask whether Eve can satisfy at least one option with probability at least the respective threshold.
Our motivation for studying DQs is twofold.
(1) DQs are interesting in their own right as they allow for a fine-grained specification of alternatives over desirable outcomes of a controlled stochastic system.
(2) DQs are equivalent to the more widely used \emph{conjunctive} queries (CQs) under an alternative semantics, namely if the opponent \Spoiler\ has to \emph{reveal} his strategy to \Achiever\ before the game starts.
This amounts to changing the quantification order over strategies from $\exists\forall$---which is the standard order---to $\forall\exists$.
In general, this makes a difference since multi-objective SGs are not \emph{determined}~\cite{cfk13}.
This optimistic ``\emph{asserted-exposure}'' ($\forall\exists$) semantics is interesting in situations where the non-deterministic system state can actually be observed at any given point in time.
For instance, in the smart heating case study from~\cite{larsen-smartheating}, the position of the doors in a house is non-deterministically controlled by \Spoiler\ and not directly observable.
However, if door sensors were to be installed, the $\forall\exists$-semantics would be more adequate since the door positions are now observable (but still uncontrollable).

\begin{figure}[t]
    \centering
    \begin{tikzpicture}[node distance=3mm and 8mm, every node/.style={scale=1}, myArrowStyle]
    \node[prob] (prob) {$s_0$};
    \node[max,above right=of prob,xshift=5mm] (s1) {$s_1$};
    \node[min,below right=of prob,xshift=5mm] (s2) {$s_2$};
    \node[prob,target,col1,right=12mm of s1] (T1) {\small $T_1$};
    \node[prob,target,col2,right=12mm of s2] (T2){\small $T_2$};
    
    \draw[trans] (prob) -- node[above left] {$\nicefrac 1 2$} (s1);
    \draw[trans] (prob) -- node[below left] {$\nicefrac 1 2$}(s2);
    \draw[trans] (s1) -- (T1);
    \draw[trans] (s1) -- (T2);
    \draw[trans] (s2) -- (T1);
    \draw[trans] (s2) -- (T2);
    
    \node[below = -1mm of s2] {\tikz[scale=.5]{
            \filldraw[black!15] (0,1) -- (1,0) -- (0,0) -- cycle;
            \draw[pcurve] (0,1) -- (1,0);
            \drawaxes
    }};
    
    \node[ above= 1mm of s1] {\tikz[scale=.5]{
            \filldraw[black!15] (0,1) -- (1,1) -- (1,0) -- (0,0) -- cycle;
            \draw[pcurve] (0,1) -- (1,1) -- (1,0);
            \drawaxes
    }};
    
    \node[left= 0mm of prob] {\tikz[scale=.5]{
            \filldraw[black!15] (0,1) -- (.5,1) -- (.5,.5) -- (1,.5) -- (1,0) -- (0,0) --cycle;
            \draw[pcurve] (0,1) -- (.5,1) -- (.5,.5) -- (1,.5) -- (1,0);
            \drawaxes
    }};

    \node[ right= 2mm of T1, yshift=1mm] {\tikz[scale=.5]{
            \filldraw[black!15] (0,1) -- (1,1) -- (1,0) -- (0,0) -- cycle;
            \draw[pcurve] (0,1) -- (1,1) -- (1,0);
            \drawaxes
    }};

    \node[ right= 2mm of T2, yshift=-1mm] {\tikz[scale=.5]{
            \filldraw[black!15] (0,1) -- (1,1) -- (1,0) -- (0,0) -- cycle;
            \draw[pcurve] (0,1) -- (1,1) -- (1,0);
            \drawaxes
    }};
    
    \end{tikzpicture}
    \caption{
        Example SG with a disjunctive query $\probability(\reach T_1) \geq x \vee \probability(\reach T_2) \geq y$.
        The Pareto sets, i.e., the feasible threshold vectors $(x,y)$ are depicted next to the respective states.    
    }
    \label{fig:introExample}
\end{figure}
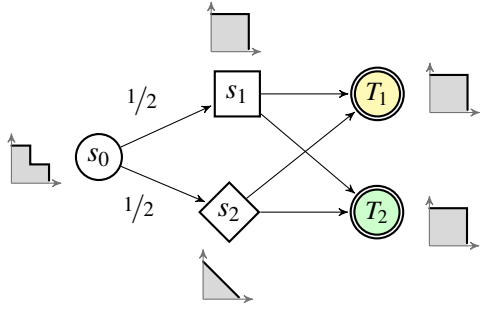

The technical intricacies of DQs are best illustrated by means of an example:
Consider the SG in Figure~\ref{fig:introExample}.
It comprises 5 states: the probabilistic state $s_0$, states $s_1$ and $s_2$ controlled by Eve and Adam, respectively, and the two targets labeled $T_1$ and $T_2$.
Suppose that Eve's objective is the DQ ``reach $T_1$ with probability at least $x$ \emph{or} $T_2$ with probability at least $y$'', in symbols $\probability(\reach T_1) \geq x \vee \probability(\reach T_2) \geq y$.
The coordinate systems next to the states show their \emph{Pareto sets}, i.e., the set of threshold vectors $(x,y)$ for which Eve can win the DQ, assuming the game starts in that state.
For $s_1$, the Pareto set is the whole box $[0,1]^2$ since Eve can reach either of the targets surely from $s_1$ by picking the respective action.
For $s_2$, the Pareto set contains all convex combinations of $(0,1)$ and $(1,0)$ and all point-wise smaller vectors, forming a triangle.
This is because Adam has to distribute the whole probability mass \emph{somewhere}; for any threshold vector $(x, y)$, $x + y \leq 1$, he cannot avoid satisfying it.
However, if $x + y > 1$, Adam can prevent Eve from winning in state $s_2$:
For example, if $(x, y) = (0.6, 0.6)$, Adam can randomize equally between both actions, and no target is reached with at least $0.6$.
This also shows that solving a DQ is \emph{not} equivalent to solving each objective separately:
From $s_2$, Eve can neither guarantee that $T_1$ nor $T_2$ is reached with positive probability; however, she can guarantee the thresholds $(0.5, 0.5)$ in the DQ.
Finally, Eve can achieve each threshold vector $(x,y)$ with $x \leq 0.5$ or $y \leq 0.5$ from $s_0$ because $s_1$ is reached with probability $0.5$ where she can put all probability mass on one of the targets.
On the other hand, Eve cannot ensure any vector $(x,y)$ with $x > 0.5 \wedge y > 0.5$ because once she has fixed a strategy, Adam can mirror it so that both targets are reached with \emph{exactly} $0.5$.

The example in Figure~\ref{fig:introExample} demonstrates two important properties of disjunctive queries in SGs:
Firstly, Pareto frontiers are not necessarily convex, as in state $s_0$.
This is in contrast to \emph{conjunctive queries} where the set of achievable probability thresholds is always convex~\cite{cfk13,EKVY08}.
Secondly, as mentioned above, SGs with multiple objectives are in general not \emph{determined}~\cite{cfk13}, i.e., it is relevant which player fixes their strategy first, thereby \emph{revealing} it to the other player before the game starts.
In fact, in the example above, switching the quantification order allows Eve to take advantage by reacting to the strategy of Adam.
For instance, she could then ensure that at least one of the targets is reached with probability at least $0.75$.

\paragraph*{Contributions and overview}
In summary, this paper makes the following contributions:
\begin{itemize}
    \item A comprehensive overview of strategy (Section~\ref{sec:strat_comp}, Table~\ref{tab:scomp}) and computational complexity (Section~\ref{sec:ccomp}, Table~\ref{tab:ccomp}) of disjunctive reachability-safety queries in stochastic games, significantly extending previous results from the literature~\cite{cfk13,FH10,RRS17}.
    In particular, motivated by the observation that randomized strategies are undesirable or meaningless for certain applications (e.g., medical or product design~\cite{quatmann}), we study the setting of DQs under \emph{deterministic} strategies for both players.
    Notably, this lead to rather high complexities:
    Qualitative queries are $\PSPACE$-hard and quantitative reachability is even undecidable.
    
   	\item A value iteration-style algorithm in the vein of~\cite{cfk13} for approximating the Pareto sets of DQs or CQs under the alternative asserted-exposure semantics (Section~\ref{sec:alg_vi}).
\end{itemize}

\paragraph*{Related work}
SGs were introduced by Shapley~\cite{shapley1953} in 1953.
Simple SGs---the turn-based variant with a reachability objective---are one of the intriguing problems in $\NP \cap \coNP$ but not known to be in $\P$~\cite{CondonCompl}.
See \cite{Condon90,gandalf} for an overview of solution algorithms.
Turn-based SGs were studied with a variety of other objectives, e.g. \cite{krishSurveyLimsupLiminf,krishVI,krishSurveyOmega}.
Other types of SGs include concurrent games~\cite{krishConcurrent,prismConcurrent}, limited information games~\cite{cav-AKW19,krishSurveyPartial}, and bidding games~\cite{surveyBidding}.

Stochastic systems with multiple objectives have been extensively studied for more than a decade.
Markov decision processes, SGs with a single player, were investigated with multiple reachability or LTL objectives~\cite{EKVY08} as well as multiple discounted sum~\cite{DBLP:conf/lpar/ChatterjeeFW13}, total reward~\cite{DBLP:conf/tacas/ForejtKNPQ11} or mean payoff objectives~\cite{DBLP:journals/lmcs/ChatterjeeKK17}.
Further, the question of percentile queries was addressed in~\cite{filar1995percentile,RRS17}, and combinations of probabilistic and non-probabilistic objectives in~\cite{bertonRaskin}.
Non-standard multi-objective queries were developed together with domain experts in~\cite{DBLP:conf/fase/BaierDKDKMW14}.

For SGs with multiple objectives, many decidability questions are still open.
For conjunctive reachability, it is only known that the Pareto set can be approximated~\cite{lics-mosg} with guaranteed precision, even in non-stopping games.
For total reward, the problem is proven decidable only for stopping games with two-dimensional queries~\cite{BF16} but it can be approximated in higher dimensional stopping SGs~\cite{cfk13}. If only deterministic strategies are allowed, the exact problem is undecidable~\cite{cfk13},
and so are generalized mean-payoff objectives in SGs~\cite{DBLP:conf/fossacs/Velner15}.
However, keeping mean-payoff above a certain threshold with some probability is $\coNP$-complete \cite{DBLP:conf/lics/Chatterjee016}. 
Further, lexicographic preferences over multiple reachability or safety objectives can be reduced to single objectives~\cite{cav-lex}.
The tool PRISM-games~\cite{prism3} implements compositional approaches to verification and strategy synthesis of several multi-objective problems~\cite{concur14basset,prism-mosg}.

To the best of our knowledge, disjunctive queries were so far only considered as a special case of more general Boolean combinations for expected rewards~\cite{cfk13}, mean-payoff SGs~\cite{DBLP:conf/fossacs/Velner15}, and in deterministic generalized reachability games on graphs~\cite{FH10}.

\iftoggle{arxiv}{}{
\paragraph*{Full version}
A full version of this paper including detailed proofs is available~\cite{arxiv}.
}

\section{Preliminaries}
\label{sec:prelims}

\subparagraph{General definitions.}
For sets $A$ and $B$, the set of functions $A \to B$ is written $B^A$.
The set of finite words over a non-empty set $A$ is written $A^*$.
For countable sets $A$ we let $\Dist(A) \eqdef \{P \in [0,1]^A \mid \sum_{a \in A} P(a) = 1\}$ be the set of all \emph{probability distributions} on $A$.
The \emph{support} of $d \in \Dist(A)$ is defined as $\supp(d) = \{a \in A \mid d(a) >0\}$.
The $i$-th component of a vector $\vec{x} \in [0,1]^n$ is denoted $x_i$.
We compare vectors $\vec{x},\vec{y} \in [0,1]^n$ component-wise, i.e., $\vec{x} \leq \vec{y}$ iff $x_i \leq y_i$ for all $i = 1,\ldots,n$.
A set $X \subseteq \Reals^n$ is \emph{convex} if for all $\vec{x}, \vec{y} \in X$ and $p \in [0,1]$ it holds that $p\vec{x} + (1{-}p)\vec{y} \in X$.
The \emph{convex hull} $\conv(X)$ of $X$ is the smallest convex superset of $X$.
Given sets $X,Y \subseteq \Reals^n$ and a real number $p \in [0,1]$, we define the $p$\emph{-convex combination} $pX + (1{-}p)Y \eqdef \{p\vec{x} + (1{-}p)\vec{y} \mid \vec{x} \in X, \vec{y} \in Y\}$.
The \emph{downward-closure} of $X \subseteq [0,1]^n$ is defined as $\dwc(X) \eqdef \{\vec{y}\in [0,1]^n \mid \exists \vec{x} \in X \colon \vec{y} \leq \vec{x}\}$.
$X$ is called \emph{downward-closed} if $X = \dwc(X)$.
A \emph{closed half-space} is a set $\{\vec{x} \in \Reals^n \mid  \vec{n}\cdot \vec{x} \geq d\}$ where $\vec{n} \in \Reals^n\setminus\{\vec{0}\}$ and $d\in \Reals$.
A \emph{polyhedron} is the intersection of finitely many closed half-spaces.
Polyhedra are convex.

\subparagraph{Stochastic games and strategies.}
Intuitively, the games considered in this paper are played by moving a pebble along the edges (called \emph{transitions} from now on) of a finite directed graph.
The vertices (subsequently called \emph{states}) of this graph are partitioned into three classes which determine the states controlled by \Achiever, \Spoiler, and the probabilistic environment, respectively:

\begin{definition}[SG]
	A \emph{stochastic game (SG)} is a tuple $\game = (\Smax,\Smin,\Sprob,\sinit,P,\choices)$, where $\Smax \uplus \Smin \uplus \Sprob \defeq S$ are finite disjoint sets of \emph{states} controlled either by \Achiever\ ($\Smax$), \Spoiler\ ($\Smin$), or the probabilistic environment ($\Sprob$). 
    The game starts in the \emph{initial state} $\sinit \in S$.
    Each $s \in \Smax \cup \Smin$ has a non-empty set $\choices(s) \subseteq S$ of \emph{actions} available to \Achiever\ (\Spoiler, resp.). For all $s \in \Sprob$, $P \colon \Sprob \to \Dist(S)$ is a probability distribution over the successors of $s$.
\end{definition}
For $s \in \Sprob$ and $t \in S$ we write $P(s,t)$ rather than $P(s)(t)$.
A state $s \in S$ is called \emph{sink} if either $s \in \Sprob$ and $P(s,s) = 1$ or $s \in S \setminus \Sprob$ and $\choices(s) = \{s\}$.
$P$ and $\choices$ together induce a directed graph on $S$.
We often sketch this \emph{game graph} in our figures (e.g.\ Figure~\ref{fig:introExample}), drawing \Achiever's, \Spoiler's and the probabilistic states with boxes, diamonds and circles, respectively; and we omit the self-loops on sinks to ease the presentation.
%
For $k \geq 0$, we let $\game^{\leq k}$ be the restriction of $\game$ to $k$-steps, that is, $\game^{\leq k}$ is obtained from $\game$ by counting the transitions from $\sinit$ taken so far and entering an error sink once their number exceeds $k$.
A \emph{Markov decision process} (MDP) is the 1-player version of an SG, i.e., either $\Smax = \emptyset$ or $\Smin = \emptyset$. A \emph{Markov chain} (MC) is a 0-player SG, i.e., $S = \Sprob$. For technical reasons, we allow MDPs and MCs with countably infinite state spaces.
SGs, on the other hand, are always \emph{finite} in this paper.

\emph{Strategies} define the semantics of SGs.
A (general) strategy for \Achiever\ is a function $\maxstrat \colon S^*\Smax \to \Dist(S)$ such that $\supp(\maxstrat(\path s)) \subseteq \choices(s)$ for all $\path s \in S^*\Smax$.
A strategy $\sigma$ is \emph{deterministic} if $\maxstrat(\path s)$ is a point-distribution for all $\path s \in S^* \Smax$, i.e., if it is not randomized.

%
To describe strategies by finite means (if possible), we use \emph{strategy automata}.
Formally, a strategy automaton for \Achiever\ is a structure $\mathcal{M} = (M, \mu, \nu,m_0)$
with $M$ a countable set of memory elements,
$\mu \colon S \times M \to M$ a \emph{memory update} function,
$\nu \colon \Smax \times M \to \Dist(S)$ a \emph{next move} function,
and $m_0 \in M$ an initial memory state.
Given a strategy automaton $\mathcal{M}$, the \emph{induced MDP} is the game
\[
    \game^\mathcal{M} ~\eqdef~ (\, \emptyset,\, \Smin \times M,\, (\Sprob \cup \Smax) \times M,\, (s_0, m_0),\, P^\mathcal{M},\, \choices^\mathcal{M} \,) ~,
\]
where the transition probability function $P^\mathcal{M}$ is defined as follows:
Let $m,m' \in M$ be arbitrary and let $s \in S$.
Then, if $s \in \Sprob$, we let $P^\mathcal{M}((s,m),(s',m')) \eqdef P(s,s')$ if $\mu(s',m) = m'$;
if $s \in \Smax$, then $P^\mathcal{M}((s,m),(s',m')) \eqdef \nu((s,m))(s')$ if $\mu(s',m) = m'$;
and $P^\mathcal{M}((s,m),(s',m')) = 0$ in all other cases.
Moreover, $\choices^\mathcal{M}((s,m)) = \{(s',m') \mid s' \in \choices(s) \wedge \mu(s',m) = m'\}$ for all $(s,m) \in \Smin \times M$.
$\mathcal{M}$ is called a \emph{stochastic-update} strategy if the memory update $\mu$ may additionally randomize over $M$ (see~\cite{DBLP:journals/corr/abs-1104-3489} for details).
From the definition it is clear that $\mathcal{M}$ \emph{realizes} a strategy in the general form $\maxstrat \colon S^*\Smax \to \Dist(S)$.
We define the \emph{memory size} of $\maxstrat$ as the smallest $k \in \Naturals \cup \{\infty\}$ such that there exists a strategy automaton $\mathcal{M} = (M, \mu, \nu,m_0)$ with $|M|=k$ that realizes $\sigma$.
If $k = \infty$, then $\maxstrat$ is \emph{infinite-memory}, and otherwise \emph{finite-memory}.
If $k=1$, then $\maxstrat$ is called \emph{memoryless}.
An \emph{MD} strategy is both memoryless and deterministic.
The above definitions are analogous for the other player \Spoiler, interchanging $\Smax$ and $\Smin$.

Throughout the paper, we consistently denote \Achiever's strategies with $\sigma$ and \Spoiler's strategies with $\minstrat$.
We usually identify strategies with their realizing automata.
Given a strategy $\maxstrat$, a \emph{counter-strategy} $\minstrat$ is a strategy in the induced MDP $\game^\maxstrat$ and may depend on $\maxstrat$.

\subparagraph{Reachability-safety queries and determinacy.}
Given an MC $(S,s_0,P)$ and a set $T \subseteq S$, the \emph{reachability probability} of $T$ is $\probability(\reach T) \eqdef \sum_{\path \in Paths(T)} Pr(\path)$ where $Paths(T)$ is the set of finite paths $\path$ of the form $\path = \sinit s_1 \ldots s_k$ with $k \geq 0$, $s_k \in T$ and $s_i \notin T$ for all $i=0,\ldots,k-1$;
and  $Pr(\path) \eqdef \prod_{i=0}^{n-1}P(s_i,s_{i+1})$.
Dually, we define $\probability(\safe T) \eqdef 1 - \probability(\reach \overline{T})$, where $\overline{T} \eqdef S \setminus T$.
Intuitively, $\probability(\reach T)$ is the probability to eventually reach a \emph{target} state in $T$ and $\probability(\safe T)$ is the probability to stay forever within $T$, i.e., to avoid the \emph{unsafe set} $\overline{T}$.
We write $\probability^{\maxstrat, \minstrat}$ to emphasize that we consider the probability measure in the Markov chain $\game^{\maxstrat, \minstrat}$ induced by some strategies $\maxstrat, \minstrat$.

%
\begin{definition}[Disjunctive Queries~\cite{cfk13}]
    Given an SG $\game$ with state space $S$, an $n$-dimensional \emph{disjunctive query (DQ)} for $\game$ is an expression of the form
    $
    \query = \bigvee_{i=1}^n \probability(\genobj_i T_i) \geq x_i 
    $
    with $\genobj_i \in \{\reach, \safe\}$, $T_i \subseteq S$ for all $i = 1,\ldots,n$, and $\vec{x} \in (0,1]^n$ a threshold vector.
\end{definition}
Note that we only allow (non-strict) \emph{lower} bounds in DQs. This is w.l.o.g.\ as upper bounds on reachability or safety can be recast as lower bounds on the dual objective.
%
%
The standard semantics of a DQ is defined as follows~\cite{cfk13}: \Achiever\ can \emph{achieve}\footnote{We use the term ``achieve'' rather than ``win'' for consistency with previous works, e.g.\ \cite{EKVY08,cfk13,BF16}.} $\query$, or equivalently, $\query$ is \emph{achievable} in $\game$ if
$
    \exists \maxstrat \forall \minstrat \colon \bigvee_{i=1}^n \probability^{\maxstrat,\minstrat}(\genobj_i T_i) \geq x_i
$
where $\maxstrat$ ranges over all strategies of \Achiever, $\minstrat$ over those of \Spoiler\ and $\probability^{\maxstrat,\minstrat}$ is the probability measure of the induced Markov chain $\game^{\maxstrat,\minstrat}$.
A strategy of \Achiever\ witnessing achievability of $\query$ is called an \emph{achieving strategy}.
%
%
Note that the quantification order is such that \Achiever\ has to reveal her strategy to \Spoiler\ \emph{before} the game actually starts.
Since the games are not determined, this may be a disadvantage (see \cite{cfk13} or our example in the introduction).
Therefore, we also consider the alternative semantics obtained by swapping the quantification order.
We call this semantics \emph{asserted-exposure} ($\forall \exists$ for short).
In the $\forall \exists$-semantics, \Spoiler's strategy is \emph{exposed} to \Achiever\ \emph{before} the game begins, i.e., $\query$ is $\forall \exists$-achievable if
$
 \forall \minstrat \exists \maxstrat \colon \bigvee_{i=1}^n \probability^{\maxstrat,\minstrat}(\genobj_i T_i) \geq x_i
$
holds.
By definition, $\game$ is \emph{determined} for $\query$ iff the standard and the alternative $\forall \exists$-semantics coincide.

We will consider the following subclasses of DQs:
If $\genobj_i = \reach$ ($\genobj_i = \safe$) for all $i=1,\ldots,n$, then we call $\query$ a \emph{reachability} (\emph{safety}) DQ.
We say that $\query$ is \emph{mixed} to emphasize that it may contain both $\reach$ and $\safe$.
If $\vec{x} = (1,\ldots,1)$ then $\query$ is called \emph{qualitative} DQ, and otherwise \emph{quantitative} DQ.
If each state contained in a target/unsafe set is a sink then $\query$ is called a \emph{sink} DQ.
All of the above notions are defined analogously for \emph{conjunctive queries} (CQs).
%

\subparagraph{Pareto sets.}
If the threshold vector $\vec{x}$ in a query $\query$ has not been fixed, we can think of $\query$ as a \emph{query template}.
We define the set \emph{Pareto set} for query template $\query$ in a given SG $\game$ as
\[
    \val(\game,\query) ~\eqdef~ \{\, \vec{x} \in [0,1]^n \,\mid\, \exists \maxstrat \, \forall \minstrat \colon \bigvee_{i=1}^n \probability^{\maxstrat,\minstrat}(\genobj_i T_i) \geq x_i \,\}
\]
and similarly for CQs.
Further, for $k \geq 0$, we call $\val(\game^{\leq k}, \query)$ the \emph{horizon-$k$} Pareto set, i.e., the set of points achievable if the game runs for at most $k$ steps. 
We define $\upval(\game, \query)$ as the set of vectors achievable in the alternative $\forall\exists$-semantics.
Note that in general, $\val(\game, \query) \subseteq \upval(\game, \query)$, and equality holds iff $\game$ is determined for all $\query(\vec{x})$, $\vec{x} \in [0,1]^n$.
The Pareto sets $\val(\game, \query)$ and $\upval(\game, \query)$ generalize the notion of \emph{lower} and \emph{upper value} in single-dimensional games.
Indeed, they coincide for $n=1$ as single-dimensional SGs are determined~\cite{CondonCompl}.
Furthermore, $\val(\game, \query)$ is convex for CQs~\cite{cfk13}.

\subparagraph{Goal-unfolding.}

The following construction is folklore (e.g.\ \cite{FH10,cfk13}).
Given an SG $\game$ with states $S$ and an $n$-dimensional query $\query$, we define the \emph{goal-unfolding} $\unf(\game,\query)$ of $\game$ with respect to $\query$ as a game with state space $S \times \{0,1\}^n$.
Let $(s, \vec{v})$ be a state of the unfolding.
Intuitively, $\unf(\game,\query)$ remembers which targets/unsafe sets have already been visited during a specific play. 
This is encoded in the $n$-bit vector $\vec{v}$.
That is, if $\game$ transitions from $s$ to $t$, then in the unfolding $(s, \vec{v})$ moves to $(t, \vec{u})$ where $\vec{u}$ is obtained from $\vec{v}$ by setting all bits corresponding to the targets/unsafe sets containing $t$ to one. Accordingly, the initial state is $(\sinit, (0,\ldots,0))$.
If $\query$ is a sink query then the unfolding is trivial, i.e., equal to $\game$.
Strategies in $\unf(\game, \query)$ can be interpreted as strategies in $\game$ by incorporating the vectors $\vec{v}$ from the states of the unfolding into the memory of the strategy.

\section{Strategy Complexity}
\label{sec:strat_comp}

In this section we analyze the memory complexity of \Achiever's achieving strategies for DQs in terms of the query dimension, denoted $n$ in the following.
More formally, given a class $Q_n$ of DQs with $n$ objectives, we determine (or bound) a number $b(n)$ such that

\begin{enumerate}
    \item $b(n)$ memory is \emph{necessary} for the queries in $Q_n$, i.e., there exists a game $\game$ and query $\query \in Q_n$ such that $\query$ is achievable in $\game$ iff \Achiever\ may use at least $b(n)$ memory;
    \item $b(n)$ memory is \emph{sufficient} for the queries in $Q_n$, i.e., for all games $\game$ and $\query \in Q_n$, if $\query$ is achievable in $\game$, then \Achiever\ has an achieving strategy using at most $b(n)$ memory.
\end{enumerate}

To give a nuanced picture of the complexity, we distinguish the classes of qualitative vs.\ quantitative DQs, safety vs.\ reachability DQs and study the restriction to deterministic vs.\ general (randomized) strategies.
We stress that the latter distinction applies to both players, i.e., 
in the \emph{deterministic-strategies} case, neither \Spoiler\ nor \Achiever\ may play randomized strategies whereas
in the \emph{general-strategies} case, both players may follow arbitrary---possibly randomized---strategies.

We briefly recall the case of non-stochastic games with deterministic strategies for both players.
It was shown in~\cite[Lem.\ 1]{FH10} that ${n \choose \lfloor n/2 \rfloor}$ memory is necessary and sufficient for safety DQs.
Notably, this means that not the \emph{whole} goal-unfolding is needed (though an exponentially large fragment).
Reachability DQs, on the other hand, do not need memory at all in the purely deterministic setting because a DQ $\bigvee_{i=1}^n \probability(\reach T_i) \geq 1$ boils down to reaching the set $\bigcup_{i=1}^n T_i$, and MD strategies are sufficient for winning reachability games on finite graphs.

In the rest of the section we first treat the general-strategies case (Section \ref{sec:strat_comp_gen}) and then study the restriction to deterministic strategies (Section \ref{sec:strat_comp_det}).
Table~\ref{tab:scomp} summarizes the results.

\begin{theorem}
    The memory bounds in Table~\ref{tab:scomp} are correct.
\end{theorem}

\begin{table}[t]
    \caption{
        Memory requirements for \Achiever\ in terms of the number $n$ of objectives in the DQ.
        The lower bounds in column \emph{``SG with deterministic strats.''} apply already to sink queries.
    }
    \label{tab:scomp}
    \centering
    {
        \setlength\tabcolsep{6pt}
        \def\arraystretch{1.1}
        \begin{tabular}{l  c  c  c  c   c c}
            \toprule
            
           & \multicolumn{2}{c}{\emph{SG with general strats.}} & \multicolumn{2}{c}{\emph{SG with deterministic strats.}} & \multicolumn{2}{c}{\emph{Non-SG with deterministic strats.}} \\
            
            & \multicolumn{2}{c}{$\safe$ /  $\reach$} & \multicolumn{2}{c}{$\safe$ /  $\reach$} & $\safe$ & $\reach$\\
            
            \midrule
            
            \emph{Qual.}   & \multicolumn{2}{c}{none~[Lem.~\ref{lem:scomp-SS-M-qual}]} & \multicolumn{2}{c}{$\geq {n \choose n/2}$ [Lem.~\ref{lem:scomp-SD-SR-qual}]} & ${n \choose \lfloor n/2 \rfloor}$~\cite{FH10} & none~[trivial] \\
            
            \midrule
            
            \emph{Quant.}  & \multicolumn{2}{c}{$2^n - 1$ [Lem.~\ref{lem:strat-upper-bound},~\ref{lem:scomp-SS-R-quant},~\ref{lem:scomp-SS-S-quant}]} & \multicolumn{2}{c}{$\infty$~[Cor.~\ref{cor:scomp:SD-SR-quant}]} & \multicolumn{2}{c}{---\emph{not applicable}---} \\
            
            \bottomrule
        \end{tabular}
    } 
\end{table}

\subsection{Strategy Complexity under General Strategies}
\label{sec:strat_comp_gen}

Recall from above that whenever we say that a certain class of strategies are ``sufficient'' or ``necessary'', we are implicitly assuming that \Achiever\ actually has a winning strategy.


\begin{restatable}{lemma}{lemstratupperbound}
    \label{lem:strat-upper-bound}
    In the general-strategies case, deterministic strategies with at most $2^n{-}1$ memory states suffice for mixed quantitative DQs. Moreover, MD strategies are sufficient for sink queries.
\end{restatable}
\begin{proof}[Proof (sketch)]
    By \cite[Theorem 7]{cfk13}, MD strategies are sufficient for quantitative DQs with \emph{expected reward} objectives (see the formal definition in~\cite{cfk13}).
    We reduce reachability and safety to expected reward in the goal-unfolding \iftoggle{arxiv}{(Appendix~\ref{app:proof-strat-upper-bound})}{\cite{arxiv}}.
    The resulting MD strategy in the unfolding corresponds to a strategy of Eve with at most $2^n{-}1$ memory.
\end{proof}

Recall that even though Lemma~\ref{lem:strat-upper-bound} implies that deterministic strategies are sufficient for \Achiever, \Spoiler\ may still use randomization.
In fact, we show in Corollary~\ref{cor:scomp:SD-SR-quant} in Section \ref{sec:strat_comp_det} that Lemma~\ref{lem:strat-upper-bound} does \emph{not} hold in the deterministic-strategies case where both players are forced to follow a deterministic strategy.
Next we either improve the bound from Lemma~\ref{lem:strat-upper-bound} or prove matching lower bounds.
We consider qualitative DQs first.

\begin{restatable}{lemma}{lemscompSSMqual}
    \label{lem:scomp-SS-M-qual}
    In the general-strategies case, MD strategies suffice for mixed qualitative DQs.
\end{restatable}
\begin{proof}[Proof (sketch)]
    It can be shown that a given strategy of \Achiever\ achieves $\bigvee_{i=1}^n \probability(\genobj_i T_i) \geq 1$ iff it achieves the single objective $\probability(\genobj_i T_i) \geq 1$ for at least one $1 \leq i \leq n$ \iftoggle{arxiv}{(Appendix~\ref{app:proof-scomp-SS-M-qual})}{\cite{arxiv}}.
    For the latter, MD strategies are sufficient~\cite{Condon90}.
\end{proof}

The qualitative bounds (``$\geq 1$'') are crucial in the previous proof.
Indeed, the equivalence in the above proof sketch does \emph{not} hold for quantitative bounds (a minimal counter-example is the game in Fig.~\ref{fig:introExample} started from $s_2$ with bounds $\geq \nicefrac{1}{2}$ for both objectives).
For the quantitative setting where arbitrary bounds are allowed, the following observation relates conjunctive and disjunctive queries and enables us to reuse some known results about CQs.
Intuitively, it states that qualitative reachability CQs can be reduced to quantitative (qualitative) DQs in the case of general (respectively deterministic) strategies.

\begin{restatable}{lemma}{lemreductioncqdq}
    \label{lem:reduction-cq-dq}
    For every SG $\game$ with a qualitative reachability CQ $\query = \bigwedge_{i=1}^n \probability(\reach T_i) \geq 1$, there exists a game $\game'$ (where \Achiever's strategies $\maxstrat$ are in one-to-one correspondence) and a
    \begin{enumerate}
        \item \emph{quantitative} reachability DQ $\query'_{quant}$ such that $\sigma$ achieves $\query$ in $\game$ iff $\sigma$ achieves $\query'_{quant}$ in $\game'$ under \emph{general} strategies;
        \item \emph{qualitative} reachability DQ $\query'_{qual}$ such that $\sigma$ achieves $\query$ iff $\sigma$ achieves $\query'_{qual}$ under \emph{deterministic} strategies.
    \end{enumerate}
\end{restatable}
\begin{proof}[Proof (sketch)]
    $\game'$ is constructed as follows:
    The original $\game$ is only played with probability $\nicefrac{1}{2}$.
    With the remaining $\nicefrac 1 2$, \Spoiler\ freely chooses one of the $n$ targets $T_1,\ldots,T_n$, after which the game ends immediately (Figure~\ref{fig:scomp-SS-S-quant}, left).
    The following can be readily verified \iftoggle{arxiv}{(Appendix~\ref{app:proof-reduction-cq-dq})}{(see~\cite{arxiv})}:
    In the deterministic-strategies case, each strategy $\maxstrat$ of \Achiever\ achieves the CQ $\query$ in $\game$ iff $\maxstrat$ achieves the 
    \emph{qualitative} DQ $\bigvee_{i=1}^n \probability(\reach T_i)\geq1 $ in $\game'$.
    Otherwise, if randomization is allowed, then $\maxstrat$ achieves $\query$ in $\game$ iff it achieves the \emph{quantitative} DQ $\bigvee_{i=1}^n \probability(\reach T_i)\geq \frac{1}{2} + \frac{1}{2n}$ in $\game'$.
\end{proof}

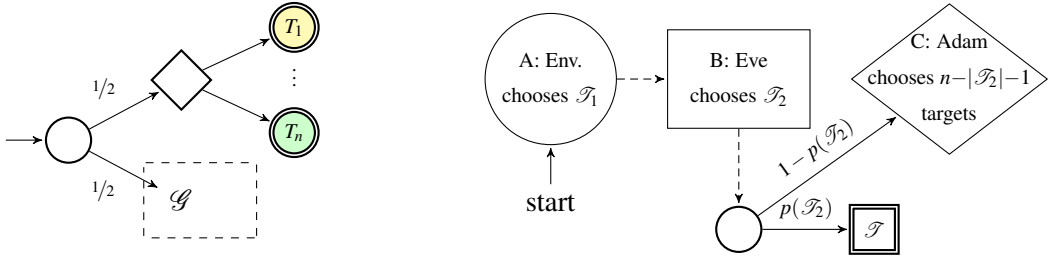
\begin{figure}[t]
    \centering
    \begin{tikzpicture}[node distance = 8mm and 15mm, on grid, initial text=,myArrowStyle]
        \node[] (gstart) {$\game$};
        
        \node[state, initial, prob, above left=of gstart, yshift=0mm] (start) {};
        \node[state, min, above right=of start, yshift=0mm] (select) {};
        \node[state, prob, target, col1, above right=of select,yshift=-1mm] (T1) {\scriptsize$T_1$};
        \node[ right=of select,yshift=1mm,scale=.7] (dots) { $\vdots$};
        \node[state, prob, target, col2, below right=of select,yshift=1mm] (T2) {\scriptsize$T_n$};
        
        \draw[trans] (start) -- node[below left]{\scriptsize $\nicefrac 1 2$} (gstart);
        \draw[trans] (start) -- node[above left]{\scriptsize$\nicefrac 1 2$} (select);
        \draw[trans] (select) -- (T1);
        \draw[trans] (select) -- (T2);
        
        \node[above=5mm of gstart, xshift=-5mm,inner sep=0] (ul) {};
        \node[above=5mm of gstart, xshift=10mm,inner sep=0] (ur) {};
        \node[below=5mm of gstart, xshift=-5mm,inner sep=0] (ll) {};
        \node[below=5mm of gstart, xshift=10mm,inner sep=0] (lr) {};
        \draw[black,dashed,-] (ul) -- (ur) -- (lr) -- (ll) -- (ul);
        
        \node[below = 14mm of start] {};
    \end{tikzpicture}
        \hspace{2cm}
    \begin{tikzpicture}[node distance=20mm and 25mm, initial text=start,on grid,initial where=below,myArrowStyle]
        \node[circle,draw,initial,align=center,inner sep=1mm] (A) {\scriptsize A: Env.\ \\ \scriptsize chooses $\targets_1$};
        \node[rectangle,draw=black,right =of A,align=center,inner sep=3mm] (B) {\scriptsize B: \Achiever\ \\ \scriptsize chooses $\targets_2$};
        \node[state,prob,below=of B] (penalty) {};
        \node[state,max,target,right=18mm of penalty] (sink) {\scriptsize $\targets$};
        \node[diamond,draw,aspect=1.3,right=28mm of B,align=center,inner sep=-7pt] (C) {\scriptsize  C: \Spoiler\ \\ \scriptsize chooses $n{-}|\targets_2|{-}1$ \\  \scriptsize targets};
        
        \draw[trans,densely dashed] (A) -- (B);
        \draw[trans,densely dashed] (B) -- (penalty);
        \draw[trans] (penalty) -- node[above] {\scriptsize$p(\targets_2)$}(sink);
        \draw[trans] (penalty) edge node[sloped,above] {\scriptsize$1- p(\targets_2)$}(C);
    \end{tikzpicture}
    \caption{
        Left:
        Reduction of qualitative CQs to DQs (Lemma~\ref{lem:reduction-cq-dq}).
        Right:
        Game constructed in the proof of Lemma~\ref{lem:scomp-SS-S-quant}.
        The probability $p(\targets_2)$ increases with $|\targets_2|$.
        \Achiever\ must select exactly $\targets_2 = \targets_1$ in Stage B to avoid visiting at least one of the targets in $\targets$ with maximal probability.
    }
    \label{fig:scomp-SS-S-quant}
\end{figure}

With Lemma~\ref{lem:reduction-cq-dq} and a result of~\cite{FH10} we obtain a lower bound for quantitative reachability:

\begin{restatable}{lemma}{lemscompSSRquant}
    \label{lem:scomp-SS-R-quant}
    In the general-strategies case, quantitative \emph{reachability} DQs need $2^n - 1$ memory.
\end{restatable}
\begin{proof}[Proof (sketch)]
    \cite[Lem. 2]{FH10} describes a family of (deterministic) games $(\game_n)_{n \geq 1}$ where \Achiever\ needs $2^n - 1$ memory to visit all $n$ targets under \emph{deterministic} strategies, i.e., to achieve a \emph{qualitative reachability CQ}.
    A more detailed analysis of the games $\game_n$ shows that even if both players may use general strategies, \Achiever\ still needs $2^{n}-1$ memory \iftoggle{arxiv}{(Appendix~\ref{app:proof-scomp-SS-R-quant})}{\cite{arxiv}}.
    We stress that this is non-trivial: in some cases, memory \emph{can} be traded for randomization to achieve qualitative DQs.
    A minimal example is an MDP with $\Smax = \{s_0\}$, $\Sprob = \{t_1,t_2\}$, $P(t_1,s_0) = P(t_2,s_0) =1$, $\choices(s_0) = \{t_1,t_2\}$, $T_1 = \{t_1\}$, $T_2 = \{t_2\}$; the CQ $\probability(\reach T_1) \geq 1 \wedge \probability(\reach T_2) \geq 1$ is achievable by a memoryless randomized strategy but not by an MD strategy.
    By Lemma~\ref{lem:reduction-cq-dq}.1, we can reduce each $\game_n$ to an SG with a \emph{quantitative reachability DQ} where \Achiever's winning strategies are the same as those in the original $\game$, i.e., they require at least $2^n{ - }1$ memory.
\end{proof}

We now prove a lower bound for safety DQs.
To the best of our knowledge, unlike Lemma~\ref{lem:scomp-SS-R-quant}, the bound for safety DQs does not easily follow from related works.
Therefore we propose a new construction for this case (shown on the right of Figure~\ref{fig:scomp-SS-S-quant}; see \iftoggle{arxiv}{Appendix~\ref{app:proof-scomp-SS-S-quant}}{\cite{arxiv}} for the proof).

\begin{restatable}{lemma}{lemscompSSSquant}
    \label{lem:scomp-SS-S-quant}
    In the general-strategies case, quantitative \emph{safety} DQs need $2^n - 1$ memory.
\end{restatable}

In summary, quantitative DQs require the \emph{full} goal-unfolding---which is of exponential size in $n$---while qualitative ones do not need memory at all.

\subsection{Strategy Complexity under Deterministic Strategies}
\label{sec:strat_comp_det}

\begin{example}
    \label{ex:det_strat_mem} 
    Consider the SG in Figure~\ref{fig:det-strat-mem-example} (left).
    The qualitative DQ $\probability(\reach T_1) \geq 1 \vee \probability(\reach T_2) \geq 1$ is achievable if \Spoiler\ has to choose deterministically in $s_0$:
    If he chooses the upper path, then \Achiever\ chooses $T_1$ in $s_1$; conversely, if he chooses the lower path, then \Achiever\ moves to $T_2$.
    Clearly, this strategy requires \Achiever\ to use one bit of memory (i.e., memory size $2$) and there is no memoryless strategy that achieves the query.
    This example shows that Lemma~\ref{lem:strat-upper-bound} is not valid in the deterministic-strategies case: memory is needed even for sink queries.
\end{example}

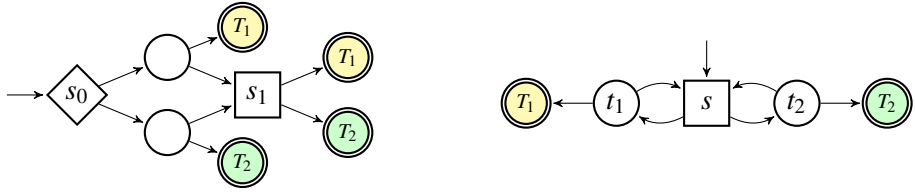
\begin{figure}[t]
    \centering
    \begin{tikzpicture}[node distance = 5mm and 12mm, on grid, initial text=,myArrowStyle]
        \node[min,initial] (s0) {$s_0$};
        \node[prob, above right=of s0] (pa) {};
        \node[prob, below right=of s0] (pb) {};
        \node[prob,target,col1,right=10mm of pa,yshift=4mm] (T1a) {\scriptsize $T_1$};
        \node[prob,target,col2,right=10mm of pb,yshift=-4mm] (T2b) {\scriptsize$T_2$};
        \node[max,below right=of pa] (s1) {$s_1$};
        \node[prob,target,col1,above right=of s1] (T1c) {\scriptsize$T_1$};
        \node[prob,target,col2,below right=of s1] (T2d) {\scriptsize $T_2$};
        
        \draw[trans] (s0) -- node[above] {} (pa);
        \draw[trans] (s0) --node[below] {} (pb);
        \draw[trans] (pa) -- (T1a);
        \draw[trans] (pb) -- (T2b);
        \draw[trans] (pa) -- (s1);
        \draw[trans] (pb) -- (s1);
        \draw[trans] (s1) --node[above] {} (T1c);
        \draw[trans] (s1) -- node[below] {}(T2d);
    \end{tikzpicture}
    \hspace{15mm}
    \begin{tikzpicture}[node distance = 5mm and 12mm, on grid, initial text=,initial where=above,myArrowStyle]
        \node[max,initial] (s) {$s$};
        \node[prob, right=of s] (s1) {$t_2$};
        \node[prob, left=of s] (s0) {$t_1$};
        \node[prob, right=of s1,target,col2] (t1) {\scriptsize$T_2$};
        \node[prob, left=of s0,target,col1] (t0) {\scriptsize$T_1$};
        
        \draw[trans] (s) edge[bend left] (s0);
        \draw[trans] (s0) edge[bend left] (s);
        \draw[trans] (s0) edge (t0);
        \draw[trans] (s) edge[bend right] (s1);
        \draw[trans] (s1) edge[bend right] (s);
        \draw[trans] (s1) edge (t1);
        
        \node[below=10mm of s] () {};
    \end{tikzpicture}
    \caption{
        All probabilities are $0.5$.
        Left:
        If only deterministic strategies are allowed, then \Achiever's achieving strategies for $\probability(\reach T_1) \geq 1 \vee \probability(\reach T_2) \geq 1$ need memory, despite the fact that targets/unsafe sets contain only sinks.
        Right:
        There exists $x$ such that the CQ $\probability(\reach T_1) \geq x \wedge \probability(\reach T_2) \geq 1{-}x$ is achievable under deterministic strategies iff infinite memory is allowed.
    }
    \label{fig:det-strat-mem-example}
\end{figure}

The above example indicates that the deterministic-strategies case is already ``interesting'' for sink queries.
Therefore---due to space limitations of the paper---we have decided to restrict our study of deterministic strategies to sink queries only.
The somewhat counter-intuitive situation that \emph{more} memory is needed if \emph{less general} strategies are allowed is due to the fact that only \Spoiler\ benefits from randomization, but not \Achiever.
In fact, \Achiever\ can in general achieve \emph{more} queries in the deterministic-strategies setting.
Nonetheless, the example in Figure~\ref{fig:det-strat-mem-example} shows that \Achiever's added power does not come for free:
She has to invest more resources (memory) on her side as well.
The following shows that this is necessary in general:

\begin{restatable}{lemma}{lemscompSDSRqual}
    \label{lem:scomp-SD-SR-qual}
    In the deterministic-strategies case, $n \choose n/2$ memory is necessary for qualitative reachability or safety DQs. This holds even for sink queries.
\end{restatable}
\begin{proof}[Proof (sketch)]
    We sketch the construction for qualitative reachability DQs. Let $n = 2q$ be even.
    The game comprises three stages, A, B, and C.
    In the initial Stage A, Adam specifies a combination $\targets_A$ of $q$ different targets.
    Those are then visited with probability $\nicefrac 1 2$ and the game ends.
    With the remaining probability of $\nicefrac 1 2$, the game moves on to Stage B which is similar to stage A, but controlled by \Achiever.
    Let $\targets_B$ be the set of $q$ targets that \Achiever\ specifies in this stage.
    Consequently, with total probability $\nicefrac 1 4$, the game enters the final stage C where Adam chooses and visits $q+1$ different targets $\targets_C$.
    It follows that a target is visited with probability $1$ iff it is chosen in all three stages.
    The only achieving strategy of \Achiever\ consists in selecting exactly $\targets_B = \targets_A$, requiring ${n \choose q}$ memory.
    See \iftoggle{arxiv}{Appendix~\ref{app:proof-scomp-SD-SR-qual}}{\cite{arxiv}} for the remaining details and the adaptation of the construction to safety.
\end{proof}

We consider quantitative DQs next.
The following lemma is the quantitative-bounds version of Lemma~\ref{lem:reduction-cq-dq}.
The difference is, however, that the reduction only works under \emph{deterministic} strategies (see~\iftoggle{arxiv}{Appendix~\ref{app:proof-reduction-cq-dq-quant}}{\cite{arxiv}} for the proof).

\begin{restatable}{lemma}{lemreductioncqdqquant}
    \label{lem:reduction-cq-dq-quant}
    For every SG $\game$ with quantitative reachability CQ $\query = \bigwedge_{i=1}^n \probability(\reach T_i) \geq x_i$, there exists a game $\game'$ (where \Achiever's strategies $\maxstrat$ are in one-to-one correspondence) and a quantitative reachability DQ $\query'$ such that, under the assumption of \emph{deterministic} strategies, $\sigma$ achieves $\query$ in $\game $ iff $\sigma$ achieves $\query'$ in $\game'$.
\end{restatable}

We can use the previous reduction to show that in general, infinite memory is necessary for achieving quantitative DQs under deterministic strategies.

\begin{lemma}
    \label{lem:mdp-inf-mem}
    For the MDP from Figure~\ref{fig:det-strat-mem-example} (right), there exists $x \in [0,1]$ such that the CQ $\probability(\reach T_1) \geq x \wedge \probability(\reach T_2) \geq 1{-}x$ is only achievable by an infinite-memory strategy.
\end{lemma}
\begin{proof}
    Every deterministic strategy $\maxstrat$ can be identified with an infinite string $\maxstrat = \maxstrat_1 \maxstrat_2 \ldots \in \{0,1\}^\omega$ where for all $i \geq 1$, $\maxstrat_i = 1$ ($\maxstrat_i = 0$) indicates that $\maxstrat$ moves to $t_1$ ($t_2$, resp.) when $s$ is entered for the $i$-th time.
    Clearly, $\probability^{\maxstrat}(\reach T_1) = (0.\maxstrat)_2$, i.e., the probability to reach $T_1$ is equal to the real number whose decimal binary representation is the infinite string $0.\maxstrat$.
    %
    %
    We claim that the above CQ $\probability(\reach T_1) \geq x \wedge \probability(\reach T_2) \geq 1{-}x$ with $x = (0.1^101^201^301^40...)_2$ can only be achieved by a strategy that uses infinite memory.
    If not, then let $\mathcal{M}$ be a finite-state strategy automaton that achieves the CQ.
    Since $\mathcal{M}$ is finite, there exist two distinct prefixes $\path_1 \neq \path_2$ of $1^101^201^301^40...$ such that after reading $\path_1$ or $\path_2$, the automaton $\mathcal{M}$ is in the same memory state $m$.
    Suppose that from $m$ on, $\mathcal{M}$ plays action sequence $\pi \in \{0,1\}^\omega$.
    Since $\mathcal{M}$ is achieving, we must have that $\pi_1 \pi = \pi_2 \pi = 1^101^201^301^40...$ which, however, implies $\pi_1 = \pi_2$, contradiction.
\end{proof}

\begin{corollary}
    \label{cor:scomp:SD-SR-quant}
    In general, infinite memory is necessary for achieving quantitative reachability or safety DQs under deterministic strategies. This holds even for sink queries.
\end{corollary}
\begin{proof}
    Apply the reduction from Lemma~\ref{lem:reduction-cq-dq-quant} to the MDP from Lemma~\ref{lem:mdp-inf-mem}.
    For safety notice that in the MDP, $\probability(\safe \overline{T}_1) \geq x$ iff $\probability(\reach T_2) \geq x$.
\end{proof}

\section{Computational Complexity}
\label{sec:ccomp}
In this section, we study the complexity of the achievability problem for DQs in the standard semantics, i.e., the decision problem ``$\exists  \maxstrat \forall \minstrat \colon \bigvee_{i=1}^n \probability^{\maxstrat,\minstrat}(\genobj_i) \geq x_i$'' in some given game.
We consider the same variations of the problem as in Section~\ref{sec:strat_comp}, that is, qualitative vs.\ quantitative DQs, reachability vs.\ safety queries and deterministic vs.\ general strategies.
For the complexity theoretic results, we assume that all transition probabilities in the games and thresholds in the queries are rational numbers given as binary-encoded integer pairs.

We again briefly discuss the case of purely deterministic games on graphs.
\cite[Theorem~1]{FH10} shows that, in deterministic games, qualitative safety DQs are $\PSPACE$-complete via a reduction from quantified Boolean formulas.
Reachability DQs, as mentioned in Section \ref{sec:strat_comp_gen}, can be reduced to solving a standard single-target reachability game which can be solved in $\P$.
We now present our results in detail, first for general (Section~\ref{sec:ccomp-gen-strat}) and then for deterministic strategies (Section~\ref{sec:ccomp-det-strat}). Table~\ref{tab:ccomp} summarizes the results.

\begin{theorem}
    The complexity bounds in Table~\ref{tab:ccomp} are correct.
\end{theorem}

\begin{table}[t]
    \caption{
        Complexity of the feasibility problem for DQs in the standard semantics.
        Column \emph{``SG with deterministic strats.''} applies exclusively to sink queries.
    }
    \label{tab:ccomp}
    \centering
    {
        \setlength\tabcolsep{4pt}
        \def\arraystretch{1.1}
    \begin{tabular}{l  c  c  c  c   c c}
        \toprule
         
          & \multicolumn{2}{c}{\emph{SG with general strats.}} & \multicolumn{2}{c}{\emph{SG with deterministic strats.}} & \multicolumn{2}{c}{\emph{Non-SG with deterministic strats.}} \\
        
         & \multicolumn{2}{c}{$\safe$ /  $\reach$} & \multicolumn{2}{c}{$\safe$ /  $\reach$} & $\safe$ & $\reach$\\
        
        \midrule
        
        \emph{Qual.}   & \multicolumn{2}{c}{$\P$~[Lem.~\ref{lem:comp-SS-qual}]} & \multicolumn{2}{c}{$\geq \PSPACE$~[Lem.~\ref{lem:comp-SD-SR-qual}] } & $\PSPACE$~\cite{FH10} & $\P$~[trivial] \\
        & & & \multicolumn{2}{c}{$\leq \EXPTIME$~[Lem.~\ref{lem:comp-SD-S-qual}] / ?} & &  \\
        
        \midrule
        
        \emph{Quant.}  & \multicolumn{2}{c}{$\geq \PSPACE$~[Lem.~\ref{lem:comp-SS-S-quant},\ref{lem:comp-SS-R-quant}]} & \multicolumn{2}{c}{undecidable~[Lem.~\ref{lem:comp-SD-R-quant}]} & \multicolumn{2}{c}{---\emph{not applicable}---} \\
        & \multicolumn{2}{c}{$\leq \NEXP$} & \multicolumn{2}{c}{} \\
 
        \bottomrule
    \end{tabular}
    } 
\end{table}

\subsection{Computational Complexity under General Strategies}
\label{sec:ccomp-gen-strat}

\begin{lemma}
    \label{lem:comp-SS-qual}
    In the general-strategies case, qualitative mixed DQs are decidable in $\P$.
\end{lemma}
\begin{proof}
    As in Lemma~\ref{lem:scomp-SS-M-qual}, in the qualitative case it suffices to check for each objective $\genobj_i T_i$, $i = 1,\ldots,n$, $\genobj_i \in \{\reach, \safe\}$, individually whether \Achiever\ can satisfy it with probability $1$.
    Hence $n$ queries to a polynomial time algorithm for qualitative simple stochastic games~\cite{DBLP:conf/stacs/EtessamiY06} suffice.
\end{proof}

\begin{lemma}
    \label{lem:comp-SS-R-quant}
    In the general-strategies case, quantitative \emph{reachability} DQs are $\PSPACE$-hard.
\end{lemma}
\begin{proof}
    The feasibility problem for qualitative reachability CQs is $\PSPACE$-hard, which holds already for MDPs where $\Smin = \emptyset$~\cite[Lem. 2]{RRS17}.
    Lemma~\ref{lem:reduction-cq-dq}.1 in Section~\ref{sec:strat_comp_gen}, reduces the MDP CQ problem to a quantitative DQ problem in an SG in polynomial time.
\end{proof}

\begin{restatable}{lemma}{lemcompSSSquant}
    \label{lem:comp-SS-S-quant}
    In the general-strategies case, quantitative \emph{safety} DQs are $\PSPACE$-hard.
\end{restatable}
\begin{proof}[Proof (sketch)]
    It can be shown that quantitative reachability CQs with \emph{strict bounds} in MDPs can be reduced to quantitative safety DQs (with non-strict bounds) in SGs.
    We prove that strict-bounded reachability CQs in MDPs are $\PSPACE$-hard which is done by analyzing a construction from \cite[Lem. 2]{RRS17} in greater detail \iftoggle{arxiv}{(see Appendix~\ref{app:comp-SS-S-quant})}{(see \cite{arxiv})}.
\end{proof}

$\PSPACE$-hardness in the previous two lemmas is caused by the exponential size of the goal-unfolding that can, as we have shown in Section~\ref{sec:strat_comp_gen}, not be avoided in general.
Indeed, for sink queries the complexity of quantitative (mixed) DQs drops to $\NP$-complete in the general-strategies case~\cite[Corollary 1]{cfk13}, where the upper bound stems from the fact that MD strategies suffice and can be verified in polynomial time using linear programming~\cite{EKVY08}.

Regarding upper bounds on quantitative DQ feasibility, we remark that the problem can be decided in $\NEXP$:
Guess an (exponentially large) MD strategy $\maxstrat$ in the goal unfolding, consider the induced MDP $\game^{\maxstrat}$ and verify in polynomial time in the size of $\game^{\maxstrat}$ that \Spoiler\ does not have a strategy violating all thresholds at once by using the multi-objective MDP algorithm from~\cite{EKVY08}.
We currently do not know of a tighter upper bound.

\subsection{Computational Complexity under Deterministic Strategies}
\label{sec:ccomp-det-strat}

As in Section~\ref{sec:strat_comp_det}, we consider only sink queries in this section.

\begin{restatable}{lemma}{lemcompSDSRqual}
    \label{lem:comp-SD-SR-qual}
    Under deterministic strategies, qualitative reachability and safety DQs are $\PSPACE$-hard, even in the case of sink queries.
\end{restatable}
\begin{proof}[Proof (sketch)]
    Inspired by similar constructions in~\cite{FH10,RRS17}, we reduce from the problem of deciding truth of a quantified Boolean formula (QBF).
    See Figure~\ref{fig:qbf_example} for an example and \iftoggle{arxiv}{Appendix~\ref{app:proof-comp-SD-SR-qual}}{\cite{arxiv}} for the full proof.
\end{proof}

\begin{figure}
    \centering
    \begin{tikzpicture}[node distance = 4mm and 12mm, on grid, initial text=,myArrowStyle]
    \node[state, max,initial] (sx) {$s_1$};
    \node[state, prob, above right=of sx] (x) {$x_1$};
    \node[state, prob, below right=of sx] (nx) {$\overline{x}_1$};
    
    \node[state, min, below right=of x] (sy) {$s_2$};
    \node[state, prob, above right=of sy] (y) {$x_2$};
    \node[state, prob,  below right=of sy] (ny) {$\overline{x}_2$};
    
    \node[state, max, below right=of y] (sz) {$s_3$};
    \node[state, prob,  above right=of sz] (z) {$x_3$};
    \node[state, prob, below right=of sz] (nz) {$\overline{x}_3$};
    
    \node[state,prob,target,col2,above right=of x] (Tx) {2};
    \node[state,prob,target,col1,above right=of y] (Ty) {1};
    \node[state,prob,target,col2,above right=of z] (Tz) {2};
    
    \node[state,prob,target,both cols,below right=of nx] (Tnx) {\scriptsize 1,2};
    \node[state,prob,col2,target,below right=of ny] (Tny) {2};
    \node[state,prob,target,col1,below right=of nz] (Tnz) {1};
    
    \draw[trans] (x) -- (Tx);
    \draw[trans] (y) -- (Ty);
    \draw[trans] (z) -- (Tz);
    \draw[trans] (z) edge[loop above] (z);
    \draw[trans] (nx) -- (Tnx);
    \draw[trans] (ny) -- (Tny);
    \draw[trans] (nz) -- (Tnz);
    \draw[trans] (nz) edge[loop below] (nz);

    \draw[trans] (sx) -- node[above left]{} (x);
    \draw[trans] (sx) -- node[below left]{} (nx);
    \draw[trans] (x) -- (sy);
    \draw[trans] (nx) -- (sy);
    
    \draw[trans] (sy) -- (y);
    \draw[trans] (sy) -- (ny);
    \draw[trans] (y) -- (sz);
    \draw[trans] (ny) -- (sz);

    \draw[trans] (sz) -- node[above left]{} (z);
    \draw[trans] (sz) -- node[below left]{} (nz);
    
    \end{tikzpicture}
    \caption{
        Example reduction from Lemma~\ref{lem:comp-SD-SR-qual} for the QBF $\exists x_1 \forall x_2 \exists x_3 (\overline{x}_1 \wedge x_2 \wedge \overline{x}_3) \vee (\overline{x}_2 \wedge x_3)$.
        Numbers in target states indicate whether they belong to $T_1$ and/or $T_2$.
        The QBF is true as witnessed by \Achiever's strategy that first goes to $\overline{x}_1$ and then to $x_3$ if \Spoiler\ had selected $\overline{x}_2$ and otherwise to $\overline{x}_3$.
    }
    \label{fig:qbf_example}
\end{figure}
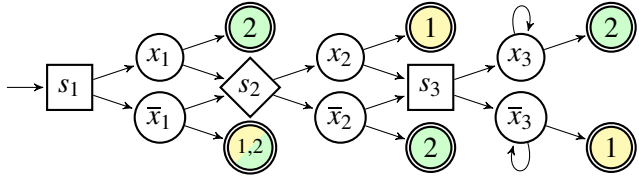

Notably, the corresponding \emph{conjunctive} problem in the setting of Lemma~\ref{lem:comp-SD-SR-qual} (qualitative sink CQ, deterministic strategies) can be solved in $\P$ as it reduces to simply checking if the intersection $\bigcap_{i=1}^n T_i$ can be reached with probability 1.
Contrary to most other results, Lemma~\ref{lem:comp-SD-SR-qual} thus identifies a setting where DQs are \emph{much} harder than CQs.
Moreover, due to the restriction to sink queries, Lemma~\ref{lem:comp-SD-SR-qual} also yields $\PSPACE$-hardness for disjunctions of expected reward objectives under deterministic strategies.
In the general-strategies case, such expected reward DQs are decidable in $\NP$~\cite{cfk13}.
Next we show an upper bound for qualitative \emph{safety} DQs:

\begin{restatable}{lemma}{lemcompSDSqual}
    \label{lem:comp-SD-S-qual}
    In the deterministic-strategies case, qualitative safety sink DQs are decidable in $\EXPTIME$.
\end{restatable}
\begin{proof}[Proof (sketch)]
    The proof relies crucially on the fact that non-achievability of a qualitative safety DQs can be witnessed after at most a bounded number steps of the game.
    Indeed, if $\maxstrat$ is a non-achieving strategy of \Achiever, then \Spoiler\ has a counter-strategy that can reach all the unsafe sets $\overline{T}$ with \emph{positive} probability after at most $|S|$ steps of the game.
    The result then follows by constructing a polynomially space-bounded alternating Turing machine that simulates the game for at most $|S|$ steps and accepts iff the query is not achievable.
    We handle probabilistic branching via backtracking using a stack whose content remains of polynomial size throughout the execution (see \iftoggle{arxiv}{Appendix~\ref{app:proof-comp-SD-S-qual}}{\cite{arxiv}} for details).
\end{proof}

The above proof cannot simply be extended to reachability because, intuitively, reaching a target with probability $1$ may only occur in the limit.
In fact, the question whether qualitative reachability DQs are decidable under deterministic strategies remains open.

Regarding the quantitative case, \cite[Theorem 3]{cfk13} proves that reachability \emph{CQs} are \emph{undecidable} under deterministic strategies.
Thus with Lemma~\ref{lem:reduction-cq-dq-quant} we also have:
\begin{lemma}
    \label{lem:comp-SD-R-quant}
    Quantitative reachability DQs are undecidable under deterministic strategies.
\end{lemma}

\section{Value Iteration}
\label{sec:alg_vi}
In a nutshell, value iteration (VI) algorithms in general evaluate the $k$-step game $\game^{\leq k}$ using information about the game $\game^{\leq k{-}1}$ up to some reasonable $k$.
In this section, we present a VI-style algorithm for computing the Pareto sets $\val(\game^{\leq k},\query)$ or $\upval(\game^{\leq k},\query')$ for a given game $\game$, DQ $\query$ or CQ $\query'$, and step bound $k \geq 0$.
Consequently, we do not fix a threshold vector $\vec{x}$ in our queries but consider query \emph{templates} instead.
To keep the presentation simple, we focus on \emph{general-strategies} and consider only sink queries.
%

We briefly recall the VI from \cite{cfk13} (subsequently called \emph{CQ-VI}) that for a given CQ $\query$ successively outputs $\val(\game^{\leq k},\query)$ for all $k \geq 0$.
Let $\pcurves$ be the set of all downward-closed polyhedra in $[0,1]^n$ and let $X \in \pcurves^S$.
In the following, we write $X_s$ for $X(s)$.
The update function $F \colon \pcurves^S \to \pcurves^S$ is defined according to Figure~\ref{fig:vi_ops}.
Note that $F$ is well-defined, i.e., always yields downward-closed polyhedra.
For CQ $\query = \bigwedge_{i=1}^n \probability(\genobj_i T_i) \geq x_i$ and each $s \in S$  define the $n$-dimensional zero-one vector $\mathbbm{1}^\query_s$, such that $(\mathbbm{1}^\query_s)_i \eqdef 1$ iff $s \in T_i$ and let $X^0(\query)_s \eqdef \dwc(\mathbbm{1}^\query_s)$.
Then \cite{cfk13} implies that $F^k(X^0(\query))_{s_0} = \val(\game^{\leq k},\query)$ for all $k \geq 0$.

\begin{figure}[t]
    \centering
    {\def\arraystretch{1.3}
        \begin{tabular}{l l l}
        \toprule
        & $F(X)_s$ & $\uviop(\uvidomelem)_s$ \\
        \midrule 
        $s \in \Smin \quad$ & $\bigcap_{t \in \choices(s)} X_t$ & $\bigcup_{t \in \choices(s)} \uvidomelem_t$ \\
        $s \in \Smax$ & $\conv\big(\bigcup_{t \in \choices(s)} X_t \big) \quad$ & $\big\{\, \conv\big( \bigcup_{t \in \choices(s)} X_t\big) \,\mid\, X \inpw \uvidomelem \, \big\}$ \\
        $s \in \Sprob$ & $\sum_{t \in S} P(s,t)X_t$ & $\big\{\, \sum_{t \in S} P(s,t)X_t \,\mid\, X \inpw \uvidomelem\, \big\}$ \\
        \bottomrule
    \end{tabular}}
    \caption{The value iteration operators $F$ and $\uviop$ for computing $\val(\game^{\leq k},\query)$ and $\upval(\game^{\leq k},\query)$, respectively.}
    \label{fig:vi_ops}
    \vspace{-0mm}
\end{figure}

We now address the question whether an iteration analogous to CQ-VI can be devised for \emph{DQs}.
CQ-VI computes the horizon-$(k{+}1)$ Pareto set of any given state by taking only the horizon-$k$ sets of its successors (and the relevant probability distribution) into account.
For DQs, this is impossible in general:

\begin{observation}
    Suppose $s \in S$ has successors $s_1, s_2$.
    In general, the horizon-$k$ Pareto sets of $s_1$ and $s_2$ w.r.t.\ a DQ do \emph{not} uniquely determine the horizon-$(k{+}1)$ Pareto set of $s$.
\end{observation}
\begin{proof}
    Consider the game in Figure \ref{fig:no_analog} (left) and the DQ $\query = \probability(\safe \overline{T_1}) \geq x_1 \vee \probability(\safe \overline{T_2}) \geq x_2$.
    The horizon-$1$ Pareto sets of $s_1, s_2$ and $t_1, t_2$, as well as the horizon-$2$ sets of $s_0$ and $t_0$ are sketched next to the corresponding state.
    We claim that the threshold vector $(x_1, x_2) = (0.75, 0.75)$ is achievable from $t_0$, but not from $s_0$:
    As deterministic strategies suffice for \Achiever, we can assume by symmetry that she moves to $T_1$ in $s_1$.
    But then \Spoiler\ can respond by moving to $T_2$ in $s_2$ and both $\safe \overline{T_1}$ and $\safe \overline{T_2}$ are satisfied with probability exactly $0.5$.
    Thus $(0.75, 0.75)$ is not achievable.
    However, at $t_2$ we can assume by symmetry that \Spoiler\ moves to $T_1$ with probability $\geq 0.5$ and so $\safe \overline{T_2}$ is satisfied with probability $ \geq 0.75$ from $t_0$.
\end{proof}

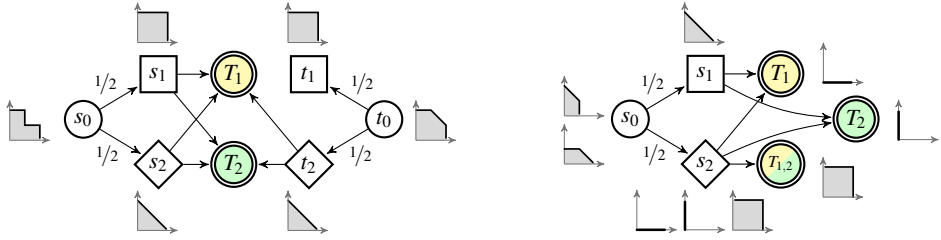
\begin{figure}[t]
    \centering
    \begin{tikzpicture}[node distance=6mm and 10mm, every node/.style={scale=0.8}, on grid, myArrowStyle]
    \node[prob] (prob) {$s_0$};
    \node[max,above right=of prob] (s1) {$s_1$};
    \node[min,below right=of prob] (s2) {$s_2$};
    \node[prob,target,col1,right=of s1] (T1) {$T_1$};
    \node[prob,target,col2,right=of s2] (T2){$T_2$};
    
    \draw[trans] (prob) -- node[above left=-1mm] {\small $\nicefrac{1}{2}$} (s1);
    \draw[trans] (prob) -- node[below left=-1mm] {\small  $\nicefrac{1}{2}$}(s2);
    \draw[trans] (s1) -- (T1);
    \draw[trans] (s1) -- (T2);
    \draw[trans] (s2) -- (T1);
    \draw[trans] (s2) -- (T2);
    
    \node[below = 6mm of s2] {\tikz[scale=.5]{
            \filldraw[black!15] (0,1) -- (1,0) -- (0,0) -- cycle;
            \draw[pcurve] (0,1) -- (1,0);
            \drawaxes
    }};
    
    \node[above= 7mm of s1] {\tikz[scale=.5]{
            \filldraw[black!15] (0,1) -- (1,1) -- (1,0) -- (0,0) -- cycle;
            \draw[pcurve] (0,1) -- (1,1) -- (1,0);
            \drawaxes
    }};
    
    \node[left= 7mm of prob] {\tikz[scale=.5]{
            \filldraw[black!15] (0,1) -- (.5,1) -- (.5,.5) -- (1,.5) -- (1,0) -- (0,0) --cycle;
            \draw[pcurve] (0,1) -- (.5,1) -- (.5,.5) -- (1,.5) -- (1,0);
            \drawaxes
    }};
    
    \node[prob, right = 40mm of prob] (prob) {$t_0$};
    \node[max,above left=of prob] (s1) {$t_1$};
    \node[min,below left=of prob] (s2) {$t_2$};
    
    \draw[trans] (prob) -- node[above right=-1mm] {\small$\nicefrac{1}{2}$} (s1);
    \draw[trans] (prob) -- node[below right=-1mm] {\small$\nicefrac{1}{2}$}(s2);
    \draw[trans] (s2) -- (T1);
    \draw[trans] (s2) -- (T2);
    
    \node[below = 6mm of s2] {\tikz[scale=.5]{
            \filldraw[black!15] (0,1) -- (1,0) -- (0,0) -- cycle;
            \draw[pcurve] (0,1) -- (1,0);
            \drawaxes
    }};
    
    \node[above= 7mm of s1] {\tikz[scale=.5]{
            \filldraw[black!15] (0,1) -- (1,1) -- (1,0) -- (0,0) -- cycle;
            \draw[pcurve] (0,1) -- (1,1) -- (1,0);
            \drawaxes
    }};
    
    \node[right= 7mm of prob] {\tikz[scale=.5]{
            \filldraw[black!15] (0,1) -- (.5,1) -- (1,.5) -- (1,0) -- (0,0) --cycle;
            \draw[pcurve] (0,1) -- (.5,1)  -- (1,.5) -- (1,0);
            \drawaxes
    }};
    \end{tikzpicture}
    \hspace{1cm}
    \begin{tikzpicture}[node distance=6mm and 10mm, on grid,every node/.style={scale=0.8}, myArrowStyle]
    \node[prob] (prob) {$s_0$};
    \node[max,above right=of prob] (s1) {$s_1$};
    \node[min,below right=of prob] (s2) {$s_2$};
    \node[prob,target,col1,right=of s1] (T1) {$T_1$};
    \node[prob,target,col2,right=30mm of prob] (T2){$T_2$};
    \node[prob,target,both cols,right=of s2] (T12){\scriptsize $T_{1,2}$};

    \draw[trans] (prob) -- node[above left=-1mm] {\small$\nicefrac{1}{2}$} (s1);
    \draw[trans] (prob) -- node[below left=-1mm] {\small$\nicefrac{1}{2}$}(s2);
    \draw[trans] (s1) -- (T1);
    \draw[trans] (s1) edge[bend right=13] (T2);
    \draw[trans] (s2) -- (T1);
    \draw[trans] (s2) edge[bend left=8] (T2);
    \draw[trans] (s2) -- (T12);
    
    \node[above =7mm of s1,xshift=0] {\tikz[scale=.5]{
            \filldraw[black!15] (0,1) -- (1,0) -- (0,0) -- cycle;
            \drawaxes
            \draw[pcurve] (0,1) -- (1,0);
    }};
    
    \node[below=6mm of s2,xshift=-8mm] {\tikz[scale=.5]{
            \drawaxes
            \draw[pcurve,very thick] (0,0) -- (1,0);
            
    }};
    \node[below=6mm of s2,xshift=0] {\tikz[scale=.5]{
            \drawaxes
            \draw[pcurve,very thick] (0,0) -- (0,1);
            
    }};
    \node[below=6mm of s2,xshift=8mm] {\tikz[scale=.5]{
            \filldraw[black!15] (0,1) -- (1,1) -- (1,0) -- (0,0) --cycle;
            \draw[pcurve] (1,0) -- (1,1) -- (0,1);
            \drawaxes
    }};
    
    \node[left=of prob,yshift=4mm,xshift=5mm] {\tikz[scale=.5]{
            
            \filldraw[black!15] (0,1) -- (.5,.5) -- (.5,0) -- (0,0) --cycle;
            \drawaxes
            \draw[pcurve] (0,1) -- (.5,.5) -- (.5,0);
    }};
    \node[left=of prob,yshift=-4mm,xshift=5mm] {\tikz[scale=.5]{
            
            \filldraw[black!15] (0,0.5) -- (.5,.5) -- (1,0) -- (0,0) --cycle;
            \drawaxes
            \draw[pcurve] (0,0.5) -- (.5,.5) -- (1,0);
    }};

    \node[right=of T1,xshift=-2mm,yshift=2mm] {\tikz[scale=.5]{
            \drawaxes
            \draw[pcurve,very thick] (0,0) -- (1,0);
            
    }};
    \node[right=of T2,xshift=-2mm] {\tikz[scale=.5]{
            \drawaxes
            \draw[pcurve,very thick] (0,0) -- (0,1);
    }};
    
    \node[right=of T12,xshift=-2mm,yshift=-2mm] {\tikz[scale=.5]{
            \filldraw[black!15] (0,1) -- (1,1) -- (1,0) -- (0,0) --cycle;
            \draw[pcurve] (1,0) -- (1,1) -- (0,1);
            \drawaxes
    }};
    \end{tikzpicture}
    \caption{Left: The successors of both $s_0$ and $t_0$ have the same Pareto set w.r.t.\ the DQ $\probability(\safe \overline{T_1}) \geq x_1 \vee \probability(\safe \overline{T_2}) \geq x_2$, but the Pareto sets of $s_0$ and $t_0$ are different. Right: Example run of Algorithm~\ref{alg:VI} for $k=2$ and query $\probability(\reach T_1) \geq x_1 \wedge \probability(\reach T_2) \geq x_2$. At $s_2$, the rightmost polyhedron is removed by the $\mu$-operation in line 4. The result is the intersection of the two polyhedra at $s_0$.}
    \label{fig:no_analog}
\end{figure}

%

Intuitively, the example in Figure~\ref{fig:no_analog} demonstrates that---unlike in CQ-VI---the Pareto sets alone do not convey enough information to allow for a sound VI.
In the remainder of this section we present a work-around for this problem.
The idea is to account for the missing information by extending the domain of VI to \emph{sets} of Pareto sets.
We will not work with DQs directly but with CQs in the $\forall\exists$-semantics.
This is justified because \emph{non-achievability} of a DQ can be recast as follows:
\begin{equation}
    \label{eq:dq-to-cq}
     \neg \, \exists \maxstrat \, \forall \minstrat \, \bigvee_{i=1}^n \probability^{\maxstrat,\minstrat}(\genobj_i T_i) \geq x_i 
    \quad \iff \quad
    \forall \maxstrat \, \exists \minstrat \, \bigwedge_{i=1}^n \probability^{\maxstrat,\minstrat}(\overline{\genobj_i} \overline{T_i}) > 1 - x_i
\end{equation}
where $\overline{\reach} = \safe$ and $\overline{\safe} = \reach$.
That is, for deciding achievability of a DQ $\query = \bigvee_{i=1}^n \probability(\genobj_i T_i) \geq x_i $, we can equivalently consider the dual CQ $\overline{\query} \eqdef \bigwedge_{i=1}^n \probability(\overline{\genobj_i} \overline{T_i}) > 1{-}x_i$ from \eqref{eq:dq-to-cq} under the $\forall\exists$-semantics in the game $\tilde{\game}$ where the roles of \Spoiler\ and \Achiever\ have been swapped.
In fact, with \eqref{eq:dq-to-cq}, the whole Pareto set $\val(\game,\query)$ can be recovered from $\upval(\tilde{\game},\overline{\query})$.

To define our VI, we introduce some auxiliary notation first.
$\pset{\pcurves}$ denotes the powerset of $\pcurves$. 
For mappings $X \in \pcurves^S$ and $\mathcal{X} \in \pset{\pcurves}^S$ from states to (sets of) polyhedra, we write $X \inpw \mathcal{X}$ iff $X(s) \in \mathcal{X}(s)$ for all $s \in S$.
Further, we let $\{X\}$ be the lifting of $X$ to $\pset{\pcurves}^S$, i.e., $\{X\}_s \eqdef \{X_s\}$.
Formally, for any fixed CQ $\query$ our new VI can be seen as a function $\uviop \colon \pset{\pcurves}^S \to \pset{\pcurves}^S$ and is defined according to Figure~\ref{fig:vi_ops}.
The iteration is started with $\mathcal{X}^0(\query) \eqdef \{X^0(\query)\}$, where $X^0(\query)$ is the same initial element as for CQ-VI.
\begin{lemma}
    \label{lem:uvicorrect}
     For all $k\geq 0$ we have that $\bigcap \left(\uviop^k(\mathcal{X}^0)_{s_0}\right) = \upval(\game^{\leq k}, \query)$.
\end{lemma}
\begin{proof}[Proof (sketch)]
    It can be shown that for all $k \geq 0$, the set $\uviop^k(\mathcal{X}^0)_{s_0}$ contains the Pareto sets achievable by \Achiever\ in the $k$-step MDPs induced by each possible \emph{deterministic} $k$-step strategy $\minstrat$ of \Spoiler\ \iftoggle{arxiv}{(Appendix~\ref{app:proof-uvicorrect})}{\cite{arxiv}}.
    The final intersection over $\uviop^k(\mathcal{X}^0)_{s_0}$ is due to the outer $\forall$-quantifier in the $\forall\exists$-semantics.
\end{proof}


The iteration according to Lemma~\ref{lem:uvicorrect} is essentially equivalent to enumerating all possible deterministic $k$-step strategies of \Spoiler\ and analyzing the induced MDPs.
In the worst case, there are doubly exponentially many (in $k$) such strategies and thus the number of polyhedra $|\uviop^k(\mathcal{X}^0)_s|$ maintained per state $s \in S$ in the $k$-th step is also at most doubly exponential.
This can be improved.
In general, not all $k$-step strategies have to be considered:
If for $k$-step  strategies $\minstrat, \minstrat'$ it holds that the set of points achievable by \Achiever\ in the induced MDP $\game^\minstrat$ is contained in set of points achievable in $\game^{\minstrat'}$, then only $\minstrat$ is relevant and $\minstrat'$ can be discarded.
Intuitively, \Spoiler\ would always (independently of the thresholds $\vec{x}$) prefer $\minstrat$ over $\minstrat'$ in such a situation.
We can incorporate this observation into our value iteration: Let $\mu \colon \pset{\pcurves}^S \to \pset{\pcurves}^S$ be the function that removes the non-inclusion-minimal polyhedra of each $\mathcal{X}_s$.
We can then iterate $\mu \circ \uviop$ instead of $\uviop$ without changing the result of the ``final intersection'' in Lemma~\ref{lem:uvicorrect}. We summarize the overall procedure as Algorithm~\ref{alg:VI}.


\begin{restatable}{theorem}{thmalgcorrect}
    \label{thm:alg-correct}
    Algorithm~\ref{alg:VI} is correct.
\end{restatable}

\subparagraph{Experiments.}

To assess the complexity of Algorithm~\ref{alg:VI} in practice we have built a prototypical implementation\footnote{Available at \url{https://doi.org/10.5281/zenodo.5047440}} using the \textsc{Parma Polyhedra Library}~\cite{PPL}.
We have tested our implementation on a variant of the smart heating example from~\cite{BF16} that was itself inspired from the case study in \cite{larsen-smartheating}.
Further, we consider randomly generated 2-dimensional games with 10 states (see \iftoggle{arxiv}{Appendix~\ref{app:exp-details}}{\cite{arxiv}} for more details).
To determine the relative overhead of our algorithm compared to CQ-VI we consider the number $n^k_s$ of polyhedra maintained at state $s$ in iteration $k$.
For the floor heating example we found that $n^k_s = 1$ for all $k \geq 0$ and $s \in S$, which means that \Spoiler\ has a unique optimal strategy from each $s$ and step-bound $k$.
Moreover, the game is determined.
For the randomly generated games, we observed that approximately 90\% of them also had $n^k_s = 1$ for all $k \geq 0$ and $s \in S$.
We conjecture that this is indeed a typical situation (as in the floor heating example), however, 90\% might be a too high estimate due to trivial random games.
In the following, we only consider ``hard'' instances with $n^k_s > 1$ for at least one state $s$ and some $k\geq 0$.
In Table 4, we report the mean number of polyhedra $\overline{n}^k = \frac{1}{|S|}\sum_{s \in S} n^k_s$ of 100 ``hard'' games for various iteration counts $k$.
The empirical average of $\overline{n}^k$ over the 100 instances that were processed within the timeout is given in column $E[\overline{n}^k]$ and the number of timeouts (10 seconds) in column T/O.
We have also compared the algorithm with and without the $\mu$-operation in Line 4 of Algorithm~\ref{alg:VI} (columns $\mu \circ \uviop$ and $\uviop$, respectively).


In summary, our experiments show that in many cases the necessary number of polyhedra is low enough to be feasible, often even only 1. In ``hard'' cases, our results show that after dozens of iterations the number of polyhedra blows up dramatically, frequently resulting in a timeout (which is why the numbers for $\uviop$ decrease after 10 iterations, as only the instances with lower $\overline{n}^k$ finish). This highlights the difficulty of DQs compared to CQs.
Still, using our optimization $\mu$, in many ``hard'' cases the computation finishes and the number of polyhedra per state stays below 10, and thus is 2 orders of magnitude smaller than without~$\mu$.


\begin{figure}[t] 
        \begin{minipage}[t]{0.6\textwidth}
            \centering
            \captionof{algorithm}{Value Iteration for CQs in the asserted-exposure ($\forall \exists$) semantics.}
            \label{alg:VI}
            \rule{\linewidth}{0.8pt}
            \begin{algorithmic}[1]
                \Require Game $\game$, (mixed) CQ $\query$, horizon $k \geq 0$
                \Ensure The horizon-$k$ Pareto set $\upval(\game^{\leq k}, \query)$
                \State $\mathcal{X} \gets \{X^0(\query)\}$ \Comment{Initialization}
                \For {$i$ from $1$ to $k$}
                \State $\mathcal{X} \gets \uviop(\mathcal{X})$ \Comment{Apply $\uviop$ according to Figure~\ref{fig:vi_ops}}
                \State $\mathcal{X} \gets \mu(\mathcal{X})$ \Comment{Keep only $\subseteq$-minima of each $\mathcal{X}_s$}
                \EndFor
                \State \textbf{return} $\bigcap \mathcal{X}_\sinit$ \Comment{Intersection of curves at initial state}
            \end{algorithmic}
            \rule{\linewidth}{0.8pt}
        \end{minipage}%
        \hfill
        \begin{minipage}[t]{0.35\textwidth}
            \centering
            \captionof{table}{Experimental results for a fixed timeout of 10s.}
            \label{tab:exp}
            {\def\arraystretch{1.099}
            \begin{tabular}{r r r r r}\toprule[0.8pt]
                 & \multicolumn{2}{c}{$E[\overline{n}^k]$} & \multicolumn{2}{c}{T/O}\\ 
                $k$ & $\uviop$  & $\mu \circ \uviop$ &$\uviop$ & $\mu \circ \uviop$\\ \midrule
                1 & 1.2 & 1.1 & 0 & 0 \\
                5 & 16.1 & 1.8 & 0  & 0\\
                10 & 526.3 & 6.8 & 63 & 12\\
                20 & 202.4 & 5.4 & 80 & 30 \\
                100 & 78.9 &2.6 & 90 & 50\\ \bottomrule[0.8pt]
            \end{tabular}}
        \vspace{0.2cm}
        \end{minipage}
\end{figure}

\section{Conclusion and Future Work}

We have presented a detailed picture of computational and strategy complexity of SGs with DQ winning conditions.
The results were obtained in part by providing reductions from CQs to DQs and applying results from the literature.
Future work on the complexity side includes closing the gaps in Tables \ref{tab:scomp} and \ref{tab:ccomp}; however, we conjecture that this requires significant new insights.
For example, a major obstacle towards proving $\PSPACE$ membership of the quantitative general-strategies DQs problem is that one has to reason about \emph{exact} reachability probabilities in the exponentially large goal-unfolding.
It is not at all obvious that the number of bits needed for the rational representations of these quantities remains polynomially bounded.

We have also argued that DQs are equivalent to CQs in the optimistic ``asserted-exposure'' ($\forall\exists$) semantics obtained by changing the quantification order over strategies---unlike in simple SGs, this makes a difference since our games are not always determined.
Moreover, we have formulated the first VI-style algorithm for DQs in the standard and CQs in the $\forall\exists$-semantics.
It should be straightforward to extend our algorithm to expected rewards as well.
Another interesting application of the algorithm is to certify determinacy (for a finite step bound).
Regarding future work, it would be appealing to implement the algorithm in a tool such as PRISM-games and to experiment with more realistic case studies.
Yet another direction is to investigate (counter-)strategy synthesis for \Achiever\ in the $\forall\exists$-semantics, e.g., by constructing strategy templates where some choices depend on \Spoiler's observable strategy.


\bibliographystyle{eptcs}
\bibliography{ref}

\iftoggle{arxiv}{
    \appendix
    

\section{Full proofs of Section~\ref{sec:strat_comp}}
\subsection{Proof of Lemma~\ref{lem:strat-upper-bound}}
\label{app:proof-strat-upper-bound}

\lemstratupperbound*

\begin{proof}
    By \cite[Theorem 7]{cfk13}, MD strategies are sufficient for quantitative DQs with \emph{expected reward} objectives (see the formal definition in~\cite{cfk13}).
    We reduce our reachability and safety objectives to expected reward objectives in the goal-unfolding.
    More formally, given a mixed DQ with $n$ objectives, each of which is either reachability or safety, we define the reward functions $r_i \colon S \to \Reals$, $i = 1,\ldots,n$, as follows:
    For each reachability objective $\reach T_i$, we grant reward $r_i(s) = 1$ if $T_i$ is visited in $s$ \emph{for the first time}, and $r_i(s) = 0$ for all other $s$.
    Similarly, for each safety objective $\safe T_i$, we grant reward $r_i(s) = -1$ if the unsafe set $S \setminus T_i$ is visited in $s$ \emph{for the first time}, and again $r_i(s) = 0$ for all other $s$.
    Then under any pair $\maxstrat, \minstrat$ of strategies, $\probability^{\maxstrat, \minstrat}(\reach T_i) \geq x_i$ iff $\expectation^{\maxstrat, \minstrat}(r_i) \geq x_i$ and $\probability^{\maxstrat, \minstrat}(\safe T_i) \geq x_i$ iff $\expectation^{\maxstrat, \minstrat}(r_i) \geq x_i - 1$, where $\expectation^{\maxstrat, \minstrat}(r_i)$ denotes the expected value of reward function $r_i$ assuming $\maxstrat, \minstrat$ are fixed.
    The result follows noticing that any strategy of \Achiever\ in the goal-unfolding can be identified with a strategy using $2^n$ memory in the original game. This can be slightly lowered to $2^n{-}1$ because \Achiever's behavior in the final stage of the unfolding where \emph{all} targets/unsafe sets have been visited is irrelevant.
\end{proof}

\subsection{Proof of Lemma~\ref{lem:scomp-SS-M-qual}}
\label{app:proof-scomp-SS-M-qual}

\lemscompSSMqual*

\begin{proof}
    We first show the following auxiliary claim: Any achieving strategy $\maxstrat$ of \Achiever\ for the qualitative DQ $\bigvee_{i=1}^n \probability(\genobj_i T_i) \geq 1$ satisfies
    \begin{equation}
    \label{eq:aux-claim}
    \exists i \in \{1,\ldots,n\}, \forall \minstrat \colon \probability^{\maxstrat, \minstrat}(\genobj_i T_i) \geq 1~.
    \end{equation}
    Towards contradiction assume \eqref{eq:aux-claim} does not hold. Then for all $T_i$, $i = 1,\ldots,n$, there exists a $\minstrat_i$ such that $\probability^{\maxstrat, \minstrat_{i}}(\genobj_i T_i) < 1$.
    Let $\minstrat$ be the strategy of \Spoiler\ that at the beginning of the game selects one of the strategies $\minstrat_i$ uniformly at random and follows this strategy for the remaining play. Note that this $\minstrat$ uses stochastic memory update. But then \emph{for each} $T_i$, $\probability^{\maxstrat, \minstrat}(\genobj_i T_i) = \frac{1}{n}\sum_{j=1}^n \probability^{\maxstrat, \minstrat_{j}}(\genobj_i T_i) \leq \frac{n-1}{n} + \frac{1}{n} \probability^{\maxstrat, \minstrat_{i}}(\genobj_i T_i) < 1$, contradiction.
    Hence by \eqref{eq:aux-claim}, \Achiever\ can play optimal for \emph{one specific} target and ignore the others. MD strategies suffice for this~\cite{Condon90}.
\end{proof}

\subsection{Proof of Lemma~\ref{lem:reduction-cq-dq}}
\label{app:proof-reduction-cq-dq}

\lemreductioncqdq*

\begin{proof}
    Let $n\geq 2$ as the claim is trivially true for $n=1$.
    We recall the construction of $\game'$ from the main text:
    The original game $\game$ is only played with probability $\nicefrac{1}{2}$.
    With the remaining $\nicefrac 1 2$, \Spoiler\ freely chooses one of the $n$ targets $T_1,\ldots,T_n$, after which the game ends immediately (Figure~\ref{fig:scomp-SS-S-quant}, left).
    If a strategy $\maxstrat$ of \Achiever\ achieves the CQ $\query$ in $\game$,
    it follows that at least one $T_i$ is visited with probability $1$ if \Spoiler\ has to resolve his extra choice deterministically.
    Otherwise, if randomization is allowed, then for all strategies $\minstrat'$ of \Spoiler\ in $\game'$ at least one $T_i$ is visited with probability $\probability^{\maxstrat,\minstrat'}(\reach T_i) \geq \frac{1}{2} + \frac{1}{2n}$.
    This means that $\maxstrat$ achieves the \emph{qualitative} DQ $\query'_{qual} \eqdef \bigvee_{i=1}^n \probability(\reach T_i)\geq1 $ in the deterministic-strategies case, and the \emph{quantitative} DQ $\query'_{quant} \eqdef \bigvee_{i=1}^n \probability(\reach T_i)\geq \frac{1}{2} + \frac{1}{2n}$ in the general-strategies case.
    On the other hand, if $\maxstrat$ does not achieve the CQ $\query$, then \Spoiler\ can ensure with some counter-strategy $\minstrat$ that at least one $T_i$, $1 \leq i \leq n$, is visited with probability $\delta < 1$ in the game $\game^\sigma$.
    In the deterministic-strategies case, he can extend $\minstrat$ to a strategy $\minstrat'$ in $\game'$ by resolving the additional choice with $T_i$.
    Then $\probability^{\maxstrat,\minstrat'}(\reach T_j) \leq \frac{1}{2}\delta + \frac{1}{2} < 1$ for all $j=1,\ldots,n$ which proves that $\maxstrat$ does not achieve the DQ $\query'_{qual}$.
    In the general-strategies case, \Spoiler\ can define $\minstrat'$ by selecting $T_i$ with probability $\frac{1}{n} + x$, where $0< x < 1{-}\delta$ and all other $T_j$, $j \neq i$, with probability $\frac{1}{n} - \frac{x}{n-1}$.
    Target $T_i$ is reached with probability equal to
    \begin{align*}
    \probability^{\maxstrat,\minstrat'}(\reach T_i) & = \frac{1}{2}\big(\frac{1}{n} + x\big) + \frac{1}{2} \delta \\
    &<  \frac{1}{2}\big(\frac{1}{n} + (1-\delta)\big)+ \frac{1}{2} \delta  \qquad \text{(as $1-\delta < 1$)}\\
    &=  \frac{1}{2n} + \frac{1}{2}~.
    \end{align*}
    Each other target $T_j$, $j \neq i$, is reached with probability at most
    \begin{align*}
    \probability^{\maxstrat,\minstrat'}(\reach T_j)& \leq \frac{1}{2}\big( \frac{1}{n} - \frac{x}{n-1}\big) + \frac{1}{2} \\
    &<  \frac{1}{2n} + \frac{1}{2}  \qquad \text{(as $x > 0$)}~.
    \end{align*}
    Thus $\probability^{\maxstrat,\minstrat'}(\reach T_j) < \frac{1}{2} + \frac{1}{2n}$ for all $j=1,\ldots,n$, implying that \Achiever's strategy $\maxstrat$ does also not achieve the quantitative DQ $\query'_{quant}$.
\end{proof}

\subsection{Proof of Lemma~\ref{lem:scomp-SS-R-quant}}
\label{app:proof-scomp-SS-R-quant}

\lemscompSSRquant*

\begin{proof}
    As explained in~\cite{FH10}, in the deterministic-strategies case \Achiever\ achieves the qualitative CQ $\bigwedge_{i=1}^n \probability(\reach T_i)\geq 1$ in the game sketched in Figure~\ref{fig:fh10-game} with the following strategy $\maxstrat$ using $2^n-1$ memory:
    She remembers exactly the set of ``petals'' that have been visited at any particular moment during a play.
    If petal $i$ is visited for the first time, she visits the corresponding target $T_i$, returning control to \Spoiler.
    If petal $i$ is visited for the second time, she plays the choice that leads to the sink, thereby visiting all targets $T_j$, $j \neq i$, and the game ends.

    First observe that this is an achieving strategy even if \Spoiler\ is allowed to use randomization.
    Assume that \Spoiler\ has a randomized counter-strategy $\minstrat$ in the finite MDP $\game^\maxstrat$.
    Then there exists some $T_i$ such that $\probability^{\maxstrat, \minstrat}(\reach T_i) < 1$.
    But since MD strategies are sufficient for minimizing reachability probabilities in finite MDPs, it follows that there exists an MD strategy $\minstrat'$ such that $\probability^{\maxstrat, \minstrat'}(\reach T_i) \leq \probability^{\maxstrat, \minstrat}(\reach T_i) < 1$.
    This is a contradiction to the fact that $\maxstrat$ reaches each target certainly under deterministic strategies.

    It remains to show that there is no randomized achieving strategy of \Achiever\ using less than $2^n - 1$ memory.
    To this end let $\maxstrat$ be a (possibly randomized) strategy of \Achiever\ using less than $2^n - 1$ memory.
    Similar to the proof in~\cite{FH10}, for each memory state $m$ of $\maxstrat$ we define the set of petals $S_m \subseteq \{1,\ldots,n\}$ as the petals where $\maxstrat$ would stop the play \emph{with positive probability} if \Spoiler\ would choose one of them next.
    Then there exists $X \subsetneq \{1,\ldots,n\}$ such that $S_m \neq X$ for all memory states $m$.
    But then \Spoiler\ wins with a strategy $\minstrat$ that always moves the token to a petal in $(X \setminus S_m ) \cup (S_m \setminus X) \neq \emptyset$:
    There are two cases:
    \begin{itemize}
        \item In each round, $\minstrat$ selects a petal $i \in X \setminus S_m$.
        Then the game runs forever with probability 1 and only the targets corresponding to the petals in $X$ are visited.
        Thus \Achiever\ loses.
        \item After some finite number of rounds, $\minstrat$ selects a petal $i \in S_m \setminus X$ for the first time.
        Then up to this point, only targets corresponding to petals in $X$ have been visited.
        In particular, $T_i$ has not been visited yet.
        \Achiever\ then loses because she moves the token with positive probability to a sink that is not contained in $T_i$.
    \end{itemize}

	To conclude the proof, we use Lemma~\ref{lem:reduction-cq-dq}.1 to reduce each $\game_n$ to an SG with a \emph{quantitative reachability DQ} where \Achiever's winning strategies are the same as those in the original $\game$, i.e., they require at least $2^n{ - }1$ memory.
\end{proof}

\begin{figure}[t]
    \centering
    \begin{tikzpicture}[node distance = 15mm]
    \node[min] (h) {};
    
    \node[max,right=of h] (e1) {};
    \node[prob,target,above = of e1] (t1) {$T_1$};
    \node[prob,target,above left= 5mm of e1,inner sep=1] (nt1) {\tiny $T_{2,3,4}$};
    
    \draw[->] (h) edge[bend right] (e1);
    \draw[->] (e1) edge[bend right] (t1);
    \draw[->] (t1) edge[bend right] (h);
    \draw[->] (e1) edge (nt1);
    
    \node[max,below=of h] (e2) {};
    \node[prob,target,right = of e2] (t2) {$T_2$};
    \node[prob,target,above right= 5mm of e2,inner sep=1] (nt2) {\tiny $T_{1,3,4}$};
    
    \draw[->] (h) edge[bend right] (e2);
    \draw[->] (e2) edge[bend right] (t2);
    \draw[->] (t2) edge[bend right] (h);
    \draw[->] (e2) edge (nt2);
    
    \node[max,left=of h] (e3) {};
    \node[prob,target,below = of e3] (t3) {$T_3$};
    \node[prob,target,below right= 5mm of e3,inner sep=1] (nt3) {\tiny $T_{1,2,4}$};
    
    \draw[->] (h) edge[bend right] (e3);
    \draw[->] (e3) edge[bend right] (t3);
    \draw[->] (t3) edge[bend right] (h);
    \draw[->] (e3) edge (nt3);
    
    \node[max,above=of h] (e4) {};
    \node[prob,target,left = of e4] (t4) {$T_4$};
    \node[prob,target,below left= 5mm of e4,inner sep=1] (nt4) {\tiny $T_{1,2,3}$};
    
    \draw[->] (h) edge[bend right] (e4);
    \draw[->] (e4) edge[bend right] (t4);
    \draw[->] (t4) edge[bend right] (h);
    \draw[->] (e4) edge (nt4);
    \end{tikzpicture}
    \caption{The (deterministic) game from~\cite[Lemma 2]{FH10} for $n = 4$ targets.}
    \label{fig:fh10-game}
\end{figure}
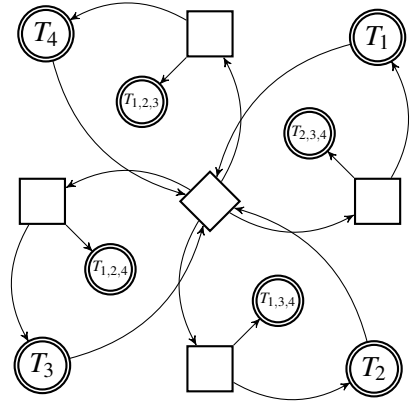

\subsection{Proof of Lemma~\ref{lem:scomp-SS-S-quant}}
\label{app:proof-scomp-SS-S-quant}

\lemscompSSSquant*

\begin{figure}[t]
	\centering
	\begin{tikzpicture}[node distance=20mm and 25mm, initial text=start,on grid,initial where=below]
		\node[circle,draw,initial,align=center,inner sep=1mm] (A) {\scriptsize A: Env.\ \\ \scriptsize chooses $\targets_1$};
		\node[rectangle,draw=black,right =of A,align=center,inner sep=3mm] (B) {\scriptsize B: \Achiever\ \\ \scriptsize chooses $\targets_2$};
		\node[state,prob,below=of B] (penalty) {};
		\node[state,max,target,right=18mm of penalty] (sink) {\scriptsize $\targets$};
		\node[diamond,draw,aspect=1.3,right=28mm of B,align=center,inner sep=-7pt] (C) {\scriptsize  C: \Spoiler\ \\ \scriptsize chooses $n{-}|\targets_2|{-}1$ \\  \scriptsize targets};
		
		\draw[trans,dashed] (A) -- (B);
		\draw[trans,dashed] (B) -- (penalty);
		\draw[trans] (penalty) -- node[above] {\scriptsize$p(\targets_2)$}(sink);
		\draw[trans] (penalty) edge node[sloped,above] {\scriptsize$1- p(\targets_2)$}(C);
	\end{tikzpicture}
	\caption{
		For convenience, we repeat here the right side of Figure~\ref{fig:scomp-SS-S-quant}, i.e., the 
		Game constructed in the proof of Lemma~\ref{lem:scomp-SS-S-quant}.
		The probability $p(\targets_2)$ increases with $|\targets_2|$.
		\Achiever\ must select exactly $\targets_2 = \targets_1$ in Stage B to avoid visiting at least one of the targets in $\targets$ with maximal probability.
	}
	\label{fig:scomp-SS-S-quant-again}
\end{figure}
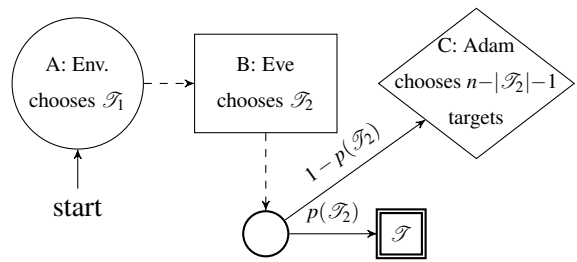

\begin{proof}
    For every $n \geq 1$ we describe a game with a safety DQ of the form $\query = \bigvee_{i=1}^n \probability(\safe \overline{T_i}) \geq x$ where achieving strategies of \Achiever\ need at least $2^n - 1$ memory.
    Note that the objectives $\probability(\safe \overline{T_i}) \geq x$ read ``avoid $T_i$ with probability at least $x$''.
    Also note that the threshold $x$ is the same for all objectives.
    Let $\mathcal{T} \eqdef \{T_1,\ldots,T_n\}$.
    The game is as follows (Figure~\ref{fig:scomp-SS-S-quant-again}):
    In Stage A, the environment chooses one of the $2^n - 1$ combinations $\targets_1 \subsetneq \targets$ uniformly at random and visits all targets in the chosen $\mathcal{T}_1$, thereby spoiling the objectives $\safe \overline{T}$ for all $T \in \targets_1$.
    In Stage B, \Achiever\ then freely selects another such combination $\targets_2 \subsetneq \targets$.
    Depending on $\targets_2$, a coin with bias $p(\targets_2)$ is tossed:
    With a certain \emph{penalty} probability $p(\targets_2)$, all targets $\targets$ are visited immediately, and with the remaining probability the game enters Stage C where \Spoiler\ chooses any probability distribution over combinations of at most $n - |\targets_2| - 1$ targets.
    The penalty $p(\targets_2)$ is strictly monotonically increasing in $|\targets_2|$, that is, selecting a larger set $\targets_2$ involves more risk from \Achiever's perspective.

    Intuitively, the idea is that one can find a threshold $x$ in $\query$, such that \Achiever\ must select \emph{exactly} $\targets_2 = \targets_1$ in Stage B.
    This clearly requires $2^n - 1$ memory as all possible combinations $\targets_1 \subsetneq \targets$ need to be remembered.
    This behavior is enforced as follows.
    There are two cases:
    (a) If $|\targets_2| < |\targets_1 \cup \targets_2|$ then $n - |\targets_2| - 1 \geq n - |\targets_1 \cup \targets_2|$ and so \Spoiler\ can visit all the remaining targets $\targets \setminus (\targets_1 \cup \targets_2)$ in Stage C, which is undesirable from \Achiever's point of view.
    (b) Otherwise $|\targets_2| \geq |\targets_1 \cup \targets_2|$ which implies $\targets_2 \supseteq \targets_1$.
    In this case, if the game reaches Stage C, at least one target will be avoided with positive probability.
    In order to offer an incentive for \Achiever\ to not simply choose an arbitrary $\targets_2 \supseteq \targets_1$, we impose a higher penalty probability on larger sets $\targets_2$.
    Then the \emph{unique} optimal strategy $\maxstrat^*$ (the one that maximizes the highest achievable threshold probability $x$) of \Achiever\ consists in choosing exactly $\targets_2 = \targets_1$.

    More formally, suppose that \Achiever\ fixes the strategy $\maxstrat^*$ described above.
    For any counter-strategy of $\minstrat$ let $X^\minstrat$ be the random variable describing the number of different visited targets.
    We can consider a counter-strategy $\minstrat$ that maximizes the expected value $E[X^\minstrat]$ (such a $\minstrat$ exists because $\game^{\maxstrat^*}$ is just a finite MDP).
    Then since the game and the strategy $\maxstrat^*$ are completely symmetric with respect to all targets, there is a maximal probability $y = E[X^\minstrat]/n$ such that \Spoiler\ can achieve $\bigwedge_{i=1}^n \probability(\reach T_i) \geq y$ in the induced MDP $\game^{\maxstrat^*}$.
    Indeed, if he could achieve $\bigwedge_{i=1}^n \probability(\reach T_i) \geq z$ for some $z > y$, then $E[X^\minstrat] \geq n  z > n  y = E[X^\minstrat]$, which is a contradiction.
    Consequently, $\maxstrat^*$ achieves the DQ $\bigvee_{i=1}^n \probability(\safe \overline{T_i}) \geq x$ with $x \eqdef 1-y$.

    On the other hand, if \Achiever\ plays any strategy $\maxstrat \neq \maxstrat^*$ we can consider each possible combination $\targets_1 \subsetneq \targets$ chosen in Stage A separately:
    If $|\targets_2| < |\targets_1 \cup \targets_2|$ (case (a) from above) then \Spoiler\ can enforce visiting \emph{all} targets with probability 1 as explained earlier.
    If $|\targets_2| \geq |\targets_1 \cup \targets_2|$ and hence $\targets_2 \supseteq \targets_1$ (case (b) from above) then nothing changes compared to playing $\maxstrat^*$ if $\maxstrat$ actually selects $\targets_2 = \targets_1$, i.e., behaves exactly like $\maxstrat^*$ given $\targets_1$ (which is possible).
    However, if $\maxstrat$ selects $\targets_2 \supsetneq \targets_1$ then for the penalties we have $p(\targets_2) > p(\targets_1)$ which means that also in this case, the probability to reach each individual target is strictly larger than if $\targets_2 = \targets_1$ had been selected by $\maxstrat$.
    Overall, \Spoiler\ can enforce $\bigwedge_{i=1}^n \probability(\reach T_i) \geq z$ for some $z > y$ because for at least one combination $\targets_1$ selected in Stage A, $\maxstrat$ differs from $\maxstrat^*$ with positive probability.
\end{proof}

\subsection{Proof of Lemma~\ref{lem:scomp-SD-SR-qual}}
\label{app:proof-scomp-SD-SR-qual}

\lemscompSDSRqual*

\begin{proof}
    We describe the construction for qualitative reachability DQs. Let $n = 2q$ be even.
    The game $\game$ comprises three stages, A, B, and C.
    In the initial stage A, Adam specifies a combination $\targets_A$ of $q$ different targets.
    Those are then visited with probability $\nicefrac 1 2$ and the game ends.
    With the remaining probability of $\nicefrac 1 2$, the game moves on to stage B which is similar to stage A, but controlled by \Achiever. Let $\targets_B$ be the set of $q$ targets that \Achiever\ specifies in this stage.
    Consequently, with total probability $\nicefrac 1 4$, the game enters the final stage C where Adam chooses and visits $q+1$ different targets $\targets_C$.
    Note that this game can be constructed such that all target states are sinks.

    Observe that a target is visited with probability $1$ iff it is chosen in all three stages.
    The only achieving strategy of \Achiever\ consists in selecting exactly $\targets_B = \targets_A$, requiring ${n \choose q}$ memory.
    On the one hand, this strategy achieves the DQ because $\targets_A \cap \targets_B \cap \targets_C \neq \emptyset$, hence there is a target that is selected in all three stages and thus visited with probability $1$.
    Note that, in agreement with Lemma~\ref{lem:strat-upper-bound}, the above argument does not work if Adam is allowed to use randomization.
    On the other hand, if $\targets_B \neq \targets_A$, then Adam can select any $q{+}1$-element subset of $\targets \setminus (\targets_A \cap \targets_B)$ in stage C, eliminating the possibility that any of the targets is selected in all three stages.

    To conclude the proof, observe that the construction can be adapted for qualitative safety DQs:
    Let $Sinks(\game)$ be the set of all sinks in $\game$.
    For each $T \in \targets$, we define a safe set $T' \eqdef T \cup (S \setminus Sinks(\game))$.
    Then for all such $T'$, the safety objective $\probability(\safe T')\geq 1$ is satisfied iff $T'$ is selected in all three stages.
\end{proof}

\subsection{Proof of Lemma~\ref{lem:reduction-cq-dq-quant}}
\label{app:proof-reduction-cq-dq-quant}

\lemreductioncqdqquant*

\begin{proof}
    The proofs parallels the one of Lemma~\ref{lem:reduction-cq-dq}.
    Let $n\geq 2$ as the claim is trivially true for $n=1$.
    We also assume that $x_i > 0$ for all $i=1,\ldots,n$.
    The game $\game'$ is the same as in Lemma~\ref{lem:reduction-cq-dq}.
    Suppose that $\maxstrat$ is a deterministic strategy of \Achiever\ that achieves $\query$ in $\game$.
    Let $\minstrat$ be any counter-strategy of \Spoiler\ in $\game'$ and suppose that it (deterministically) resolves the extra choice by moving to $T_i$, for some $i$.
    Then $\probability^{\maxstrat,\minstrat}(\reach T_i) \geq \frac{1}{2}x_i + \frac{1}{2}$ as $\maxstrat$ achieves $\query$ in the original $\game$.
    Thus $\maxstrat$ achieves the DQ $\query' \eqdef \bigvee_{i=1}^n \probability(\reach T_i) \geq \frac{1}{2}x_i + \frac{1}{2}$ in $\game'$.

    Conversely, if $\maxstrat$ does not achieve $\query$ in $\game$, then there exists some $1 \leq i \leq n$ such that \Spoiler\ has a counter-strategy $\minstrat$ satisfying $\probability^{\maxstrat,\minstrat}(\reach T_i) < x_i$.
    We extend this $\minstrat$ to a strategy $\minstrat'$ in $\game'$ by selecting $T_i$ at the additional state.
    Then in $\game'$, $\probability^{\maxstrat,\minstrat'}(\reach T_i) < \frac{1}{2}x_i + \frac{1}{2}$ and $\probability^{\maxstrat,\minstrat'}(\reach T_j) \leq  \frac{1}{2} < \frac{1}{2}x_j + \frac{1}{2}$ for all other targets $j \neq i$.
    Therefore $\maxstrat$ does not achieve $\query'$ in $\game'$ and the proof is complete.
\end{proof}

\section{Full proofs of Section~\ref{sec:ccomp}}
\subsection{Proof of Lemma~\ref{lem:comp-SS-S-quant}}
\label{app:comp-SS-S-quant}

\lemcompSSSquant*

\begin{proof}
    We show that the complement of the quantitative safety DQ feasibility problem is $\PSPACE$-hard which proves our result as $\PSPACE$ is closed under complement.
    Consider an MDP $\game$ with $\Smax = \emptyset$.
    Let $\game'$ be like the MDP $\game$ but with $\Smax$ and $\Smin$ exchanged.
    For any quantitative safety DQ $\query = \bigvee_{i=1}^n\probability(\safe T_i) \geq x_i$ it holds that
    \begin{align*}
    & \text{$\query$ is not achievable in $\game$} \\
    \iff & \neg \exists \maxstrat \forall \minstrat \colon \bigvee_{i=1}^n\probability^{\maxstrat,\minstrat}(\safe T_i) \geq x_i \\
    \iff & \forall \maxstrat \exists \minstrat \colon \bigwedge_{i=1}^n\probability^{\maxstrat,\minstrat}(\safe T_i) < x_i \\
    \iff & \forall \maxstrat \exists \minstrat \colon \bigwedge_{i=1}^n\probability^{\maxstrat,\minstrat}(\reach \overline{T_i}) > 1{-}x_i \\
    \iff & \exists \minstrat \colon \bigwedge_{i=1}^n\probability^{\minstrat}(\reach \overline{T_i}) > 1{-}x_i \tag{because $\Smax = \emptyset$} \\
    \iff & \text{the CQ $\bigwedge_{i=1}^n\probability(\reach \overline{T_i}) > 1{-}x_i$ with \emph{strict bounds} is achievable in $\game'$}
    \end{align*}
    It only remains to show that the quantitative reachability CQ feasibility problem with \emph{strict bounds} is $\PSPACE$-hard in games with $\Smin = \emptyset$, i.e., MDPs where only \Achiever\ has choices.
    To this end we analyze the reduction from \cite[Lem.\ 2]{RRS17} (restated in Figure~\ref{fig:mdp_rrs14}) which shows that the qualitative reachability CQ feasibility problem in MDPs is $\PSPACE$-hard in greater detail.

    \begin{figure}[t]
        \centering
        \begin{tikzpicture}[node distance = 4mm and 15mm, on grid, initial text=]
        \node[state, max,initial] (sx) {$x_1$};
        \node[state, prob, above right=of sx] (x) {$f_1$};
        \node[state, prob, below right=of sx] (nx) {$t_1$};
        
        \node[state, prob, below right=of x] (sy) {$x_2$};
        \node[state, prob, above right=of sy] (y) {$f_2$};
        \node[state, prob,  below right=of sy] (ny) {$t_2$};
        
        \node[ below right=of y] (dots) {$\cdots$};
        
        \node[state, max, right=of dots] (sz) {$x_n$};
        \node[state, prob,  above right=of sz] (z) {$f_n$};
        \node[state, prob, below right=of sz] (nz) {$t_n$};
        
        \draw[trans] (z) edge[loop right] (z);
        
        \draw[trans] (nz) edge[loop right] (nz);

        \draw[trans] (sx) -- node[above left]{} (x);
        \draw[trans] (sx) -- node[below left]{} (nx);
        \draw[trans] (x) -- (sy);
        \draw[trans] (nx) -- (sy);
        
        \draw[trans] (sy) -- node[above]{$\nicefrac 1 2$} (y);
        \draw[trans] (sy) -- node[below]{$\nicefrac 1 2$} (ny);
        \draw[trans] (y) -- (dots);
        \draw[trans] (ny) -- (dots);

        \draw[trans] (sz) -- node[above left]{} (z);
        \draw[trans] (sz) -- node[below left]{} (nz);
        
        \end{tikzpicture}
        \caption{
            The MDP from \cite[Lem.\ 2]{RRS17} constructed from a quantified Boolean formula (QBF) $\exists x_1 \forall x_2 \exists x_3 \ldots \forall x_{m-1} \exists x_m \colon C_1 \wedge C_2 \wedge \ldots \wedge C_n$ where each clause $C_i$ is a disjunction of three literals from the set $\{x_i,\, \overline{x_i} \mid i \in \{1,\ldots,m\}\}$.
            The target sets are $T_i \eqdef \{t_j \mid x_j \in C_i\} \cup \{f_j \mid \overline{x_j} \in C_i\}$ for all $i \in \{1,\ldots,n\}$ and $j \in \{1,\ldots,m\}$.
            It was proved in \cite{RRS17} that the QBF is true iff the qualitative reachability CQ $\bigwedge_{i=1}^m\probability(\reach T_i) \geq 1$ is achievable.
        }
        \label{fig:mdp_rrs14}
    \end{figure}
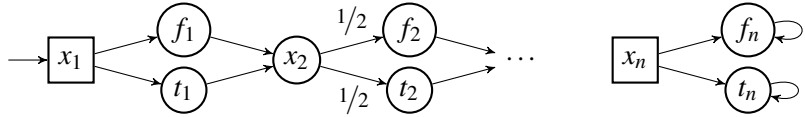

    We claim that the QBF is true iff the strict-bounds query $\query_{>} \eqdef \bigwedge_{i=1}^m\probability(\reach T_i) > 1 - \frac{2^{-n}}{m}$ is achievable in the constructed MDP.

    It was already shown in the proof of \cite[Lem.\ 2]{RRS17} that if the QBF is true, then even the qualitative query $\bigwedge_{i=1}^m\probability(\reach T_i) \geq 1$ is achievable, and thus also our quantitative strict-bounds query $\query_{>}$.

    For the other direction assume that the QBF is false.
    Then for all $\exists$-strategies in the ``formula game'', there is a counter-strategy of the $\forall$-player such that at least one clause is not satisfied.
    Let $\maxstrat$ be an arbitrary (possibly randomized) strategy of \Achiever\ in the constructed MDP.
    %
    %
    We overload $\reach T$ and interpret it as the \emph{event} that $T$ is reached.
    Formally, $\reach T$ can still be regarded a \emph{set} of paths.
    Thus the event $\bigcap_{i=1}^m \reach T_i$ indicates that \emph{all} $m$ targets are reached, and $\probability^{\maxstrat}(\bigcap_{i=1}^m \reach T_i)$ is the probability of that event.
    We have
    \begin{equation}
    \label{eq:mdp_strict_bounds}
    \probability^{\maxstrat}(\bigcap_{i=1}^m \reach T_i) \leq 1 - 2^{-n}
    \end{equation}
    because no matter the behavior of $\maxstrat$, there is at least a probability of $2^{-n}$ that the environment makes the ``right'' choices that make the formula false.
    In the MDP this means that at least one target is not visited with probability at least $2^{-n}$ or, equivalently, that all targets are visited with probability at most $1-2^{-n}$.
    The following inequalities complete the proof:
    \begin{align*}
    & \probability^{\maxstrat}(\bigcap_{i=1}^m \reach T_i) \\
    = & 1 - \probability^{\maxstrat}(\bigcup_{i=1}^m \safe \overline{T_i})\tag{probability of complement}\\
    \geq & 1 - \sum_{i=1}^m \probability^{\maxstrat}(\safe \overline{T_i}) \tag{union bound}\\
    = & \sum_{i=1}^m (1 - \probability^{\maxstrat}(\safe \overline{T_i})) - (m-1) \tag{arithmetic}\\
    = & \sum_{i=1}^m \probability^{\maxstrat}(\reach T_i) - (m-1) \tag{probability of complement}\\
    \end{align*}
    and thus with \eqref{eq:mdp_strict_bounds} we obtain
    \[
    \sum_{i=1}^m \probability^{\maxstrat}(\reach T_i) \leq \probability^{\maxstrat}(\bigcap_{i=1}^m \reach T_i) + m-1 \leq 1 - 2^{-n} + m-1 = m - 2^{-n}~.
    \]
    From this it follows that
    \[
    \min_{i = 1,\ldots,m} \probability^{\maxstrat}(\reach T_i) \leq \frac{m - 2^{-n}}{m} = 1 - \frac{2^{-n}}{m} < 1~.
    \]
    Thus since $\maxstrat$ was arbitrary, we have shown that $\forall \maxstrat \bigvee_{i=1}^m\probability^{\maxstrat}(\reach T_i) \leq  1 - \frac{2^{-n}}{m}$ which is equivalent to saying that the query $\query_{>}$ defined above is not achievable.
\end{proof}

\subsection{Proof of Lemma~\ref{lem:comp-SD-SR-qual}}
\label{app:proof-comp-SD-SR-qual}

\lemcompSDSRqual*

\begin{proof}
    We reduce from the problem of deciding truth of a quantified Boolean formula (QBF). A QBF is an expression $Q_1 x_1 Q_2 x_2 \ldots Q_m x_m \colon\psi$ with $Q_i \in \{\exists,\forall\}$ for all $i = 1,\ldots,m$ and where $\psi$ is a Boolean formula with variables $x_1,\ldots,x_m$. Deciding truth of a QBF is known to be $\PSPACE$-hard already when $\psi$ is in CNF. Consequently, since $\PSPACE$ is closed under complement, the problem is $\PSPACE$-hard for $\psi$ in DNF as well. 
    
    Given a QBF with $\psi = t_1 \vee \ldots \vee t_n$ in DNF, we construct the following game $\game_{\psi}$ (see Figure~\ref{fig:qbf_example} for an example):
    There is a target $T_j$ for each term $t_j$, $j=1,\ldots,n$, of $\psi$. Moreover, there is a state $s_i$ for each variable $x_i$ such that $s_i \in \Smax$ if $Q_i = \exists$ and $s_i \in \Smin$ if $Q_i = \forall$.
    For all $i = 1,\ldots,m$, $s_i$ has exactly two successors $x_i, \overline{x_i} \in \Sprob$.
    Each $x_i$ (resp.\ $\overline{x_i}$) leads with probability $\nicefrac 1 2$ to a state where exactly the targets $T_j$ such that either (i) the literal $x_i$ (resp.\ $\overline{x_i}$) occurs in $t_j$ or (ii) the variable $x_i$ does not occur at all in $t_j$ are satisfied. 
    With the remaining probability of $\nicefrac 1 2$, the game moves on the next $s_{i+1}$ (except if $i=m$, in which case the game stays in $x_m$, respectively $\overline{x_m}$).
    Clearly, this game can be constructed in polynomial time in the size of the QBF.
    
    As explained in~\cite{FH10}, the deterministic strategies in $\game_{\psi}$ are in one-to-one correspondence with the ``strategies'' assigning the values in the \emph{formula game} $Q_1 x_1 Q_2 x_2 \ldots Q_m x_m \colon\psi$. Note that the increase in quantifier alternations (up to $m-1$ many) in the formula game compared to the standard DQ semantics (one alternation) is because in an SG, players can remember the whole history and thus react to the opponent's past choices.
    
    In the game $\game_{\psi}$, for each (deterministic) strategy pair $\maxstrat,\minstrat$, it holds that $\bigvee_{j=1}^n\probability^{\maxstrat,\minstrat}(\reach T_j) \geq 1$ iff there is at least one term $t_j$ such that $\maxstrat$ and $\minstrat$ never select a literal whose negation occurs in $t_j$ iff the variable assignment defined by $\maxstrat,\minstrat$ satisfies $\psi$. Therefore, an achieving strategy of \Achiever\ witnesses truth of the QBF. Conversely, if the QBF is true, then we can map a winning strategy of the $\exists$-player in the formula game to a deterministic strategy in $\game_{\psi}$, such that for all counter-strategies of \Spoiler, there exists a term $t_j$ such that no player ever selects a literal whose negation is contained in $t_j$. This leads to reaching $T_j$ with probability $1$.
    
    We can adapt the construction to qualitative safety DQs similar as in Lemma~\ref{lem:scomp-SD-SR-qual}.
\end{proof}

\subsection{Proof of Lemma~\ref{lem:comp-SD-S-qual}}
\label{app:proof-comp-SD-S-qual}

\lemcompSDSqual*

\begin{proof}
    Recall that we have to decide the following problem, where $\maxstrat$ and $\minstrat$ range over \emph{deterministic} strategies only and $\probability$ implicitly refers to an SG $\game$:
    $\exists\maxstrat \forall \minstrat \bigvee_{i=1}^n \probability^{\maxstrat,\minstrat} (\safe T_i) \geq 1$.
    We show how to decide the complement of this problem with a polynomially space-bounded alternating Turing machine (ATM).
    This implies the claimed result as alternating $\PSPACE$ equals $\EXPTIME$~\cite{DBLP:journals/jacm/ChandraKS81}.
    Note that the complement is the decision problem
    $\forall \maxstrat \exists \minstrat \bigwedge_{i=1}^n \probability^{\maxstrat,\minstrat}(\reach \overline{T_i}) > 0$ where each $\overline{T_i}$ occurs only in sinks.
    We need the following observation:
    \begin{claim*} If all $T_i$, $i = 1,\ldots,n$, contain only sinks, then
        \[
        \forall \maxstrat \exists \minstrat \bigwedge_{i=1}^n \probability^{\maxstrat,\minstrat}(\reach T_i) > 0 \iff \forall \maxstrat \exists \minstrat \bigwedge_{i=1}^n \probability^{\maxstrat,\minstrat}(\reach^{\leq |S|} T_i) > 0
        \]
        where $\probability^{\maxstrat,\minstrat}(\reach^{\leq |S|} T_i)$ denotes the probability to reach $T_i$ in the $|S|$-step game $\game^{\leq |S|}$ under strategies $\maxstrat,\minstrat$.
    \end{claim*}
    \begin{claimproof}
        The direction from right to left is trivial.
        For the other direction observe that if there exists $\maxstrat$ of \Achiever\, such that \Spoiler\ cannot reach one of the targets $T_i$ for some $1 \leq i \leq n$ within $|S|$ steps, then \Spoiler\ cannot reach $T_i$ in any number of steps.
        Indeed, as targets only need to be reached with positive probability, we can assume that the probabilistic states behave in \Spoiler's favor.
        In other words, if \Spoiler\ cannot reach $T_i$ for some $i$ in $|S|$ steps against a given strategy of \Achiever, then he can also not reach it even if he controlled all the random states and so he cannot reach it in any number of steps.
    \end{claimproof}
    
    For the rest of the proof we assume w.l.o.g.\ that all states $s \in \Sprob$ have exactly two successors, called left and right in the following (e.g.~\cite[Appendix A]{gandalf}).

    The ATM simulates the game for exactly $|S|$ steps (a similar technique was used in \cite{FH10}).
    States controlled by \Achiever\ correspond to universal states of the ATM and those controlled by \Spoiler\ to existential ones (note that since we consider the complementary problem, strategies of \Achiever\ are universally quantified and those of \Spoiler\ existentially).
    This can be implemented in the ATM by keeping a record of the current state of the game.
    The ATM also keeps track of the history $\path \in S^*$ of the current path and stores a bit vector $(t_1,\ldots,t_n) \in \{0,1\}^n$ indicating which targets have already been reached.
    In order to handle probabilistic states, the ATM maintains an additional stack.
    When the game reaches a probabilistic state which is not a sink, the ATM pushes the current history on the stack and explores the left branch first.
    Once the ATM reaches a configuration where the current history exceeds $|S|$ and the stack is not empty, it pops the saved history $s_0\ldots s$, $s \in \Sprob$ from the stack, sets its current history to $s_0\ldots s$ and the current game state to $s$ and explores the right successor of $s$.
    If the ATM reaches a configuration where the current history exceeds length $|S|$ and the stack is empty, then it accepts iff $(t_1,\ldots,t_n) = (1,\ldots,1)$.

    Overall, the ATM accepts iff $\forall \maxstrat \exists \minstrat \bigwedge_{i=1}^n \probability^{\maxstrat,\minstrat}(\reach \overline{T_i}) > 0$ holds:
    \begin{itemize}
        \item If the ATM accepts, then \Spoiler\ can reach each $\overline{T_i}$, $i=1,\ldots,n$, already within $|S|$ steps with positive probability against all strategies of \Achiever.
        Thus $\probability^{\maxstrat,\minstrat}(\reach \overline{T_i}) > 0$ holds.
        \item Conversely, if $\forall \maxstrat \exists \minstrat \bigwedge_{i=1}^n\probability^{\maxstrat,\minstrat}(\reach \overline{T_i}) > 0$ holds, then we also have $\forall \maxstrat \exists \minstrat \bigwedge_{i=1}^n \probability^{\maxstrat,\minstrat}(\reach^{\leq |S|} \overline{T_i}) > 0$ by the auxiliary claim.
        Thus the ATM accepts since it simulates the game appropriately for $|S|$ steps.
    \end{itemize}
\end{proof}

\section{Full proofs of Section~\ref{sec:alg_vi}}
\subsection{Proof of Lemma~\ref{lem:uvicorrect}}
\label{app:proof-uvicorrect}

    First, we only have to consider \emph{deterministic} strategies of \Spoiler\ because from his point of view, he is trying to achieve a (mixed) \emph{DQ} and deterministic strategies are sufficient for this by Lemma~\ref{lem:strat-upper-bound}.
    A deterministic \emph{$k$-step strategy} of \Spoiler\ is a deterministic strategy in the $k$-step game $\game^{\leq k}$.
    Formally, such a strategy can be identified with a mapping
    \[
        \minstrat^k \colon S^{\leq k-1} \Smin \to S
    \]
    such that $\minstrat^k(\path s) \in \choices(s)$ for all $\path s \in S^{\leq k-1} \Smin$.
    Here, $S^{\leq k-1}$ denotes the set of paths of length at most $k-1$ (by definition, $S^{-1} = \emptyset$ and $S^0 = \{\epsilon\}$, where $\epsilon$ is the \enquote{empty path}).
    In particular, $\minstrat^k$ needs to take only histories of length at most $k$ into account and may not randomize.
    In the following, we write $\mathcal{X}^0 = \mathcal{X}^0(\query)$.
    We will show the following invariant:
    \begin{quote}
        For all $k \geq 0$ and $s \in S$, the set $\uviop^k(\mathcal{X}^0)_s$ contains a polyhedron $X^k_s(\minstrat) \in \pcurves$ for each possible \emph{deterministic} $k$-step strategy $\minstrat$ of \Spoiler.
        This polyhedron satisfies $X^k_s(\minstrat) = \val(\game_s^{\leq k, \minstrat}, \query)$ where $\game_s$ is like $\game$ with initial state $s$.
        In words, $X^k_s(\minstrat)$ is precisely the Pareto set achievable from state $s$ by \Achiever\ in the induced $k$-step MDP $\game^{\leq k, \minstrat}$.
    \end{quote}
    The above invariant establishes the claim as the universal quantifier over \Spoiler's strategies translates to intersection over the polyhedra $X^k_s(\minstrat)$.
    The invariant is proved by induction on $k$.
    For $k = 0$, it holds trivially.
    Now let $k > 0$.
    First observe that a deterministic $k$-step strategy $\minstrat^k$ of \Spoiler\ can be decomposed as follows:
    For all $s \path t \in S S^{\leq k{-}1} \Smin$,
    \[
        \minstrat(s \path t) = \minstrat_s^{k-1}(\path t)
    \]
    for appropriate $(k{-}1)$-step strategies $\minstrat_s^{k-1}$ for each $s \in S$.
    We now prove existence of a polyhedron $X^k_s(\minstrat^k) \in \uviop^k(\mathcal{X}^0)_s$ for each $s \in S$ as required by the invariant.
    Therefore let $\minstrat^k$ be an arbitrary deterministic $k$-step strategy and let $s \in S$ be an arbitrary state.
    If $s$ is a sink, then the invariant is trivially satisfied due to the choice of the initial element $\mathcal{X}^0$.
    Thus assume that $s$ is not a sink.
    Then since $\query$ is a sink query, $s$ is not contained in any target/unsafe set and so we can ignore this case in the following.
    Let $\minstrat^{k-1}_{s'}$ be the decomposition of $\minstrat^k$ into $(k{-}1)$-step strategies as above, one for each $s' \in S$.
    %
    %
    We make a case distinction.
    \begin{itemize}
        \item $s \in \Smin$.
        By the induction hypothesis (IH), $X^{k-1}_t(\minstrat_t^{k-1}) \in \uviop^{k-1}(\mathcal{X}^0)_t$ for all $t\in S$.
        Let $\minstrat^k(s) = r \in \choices(s)$, i.e., $r$ is the successor of $s$ if \Spoiler\ follows $\minstrat^k$ and the game starts in $s$.
        Then $X^{k}_s(\minstrat^{k}) = X^{k-1}_{r}(\minstrat_r^{k-1})$ and thus $X^{k}_s(\minstrat^{k}) \in \uviop^{k-1}(\mathcal{X}^0)_r$ by the IH.
        By definition, $\uviop^{k}(\mathcal{X}^0)_s = \bigcup_{t \in \choices(s)} \uviop^{k-1}(\mathcal{X}^0)_t$ and thus $X^{k}_s(\minstrat^{k}) \in \uviop^{k}(\mathcal{X}^0)_s$ since $r \in \choices(s)$.
        \item $s \in \Smax$.
        %
        %
        %
        We have $X^{k}_s(\minstrat^{k}) = \conv(\bigcup_{t \in \choices(s)} X^{k-1}(\minstrat^{k-1}_{t}))$, where the convex hull reflects the fact that \Achiever\ may randomize at $s$ and therefore she cannot only ensure the points in $X^{k-1}(\minstrat^{k-1}_{t})$ but also all their convex combinations.
        By the IH, $X^{k-1}_t(\minstrat_t^{k-1}) \in \uviop^{k-1}(\mathcal{X}^0)_t$ for all $t\in S$.
        By definition, $\uviop^{k}(\mathcal{X}^0)_s = \{\conv(\bigcup_{t \in \choices(s)} X_t \mid X \inpw\uviop^{k-1}(\mathcal{X}^0) \}$,
        and thus $X^{k}_s(\minstrat^{k}) \in \uviop^k(\mathcal{X}^0)_s$.
        \item $s \in \Sprob$.
        In this case, $X^{k}_s(\minstrat^{k}) = \sum_{t \in S}P(s,t) X^{k-1}_t(\minstrat_t^{k-1})$.
        Similar to the previous case, we have $X^{k-1}_t(\minstrat_t^{k-1}) \in \uviop^{k-1}(\mathcal{X}^0)_t$ for all $t\in S$ by the IH
        and by definition, $\uviop^{k}(\mathcal{X}^0)_s = \{\sum_{t \in S} X_t \mid X \inpw\uviop^{k-1}(\mathcal{X}^0) \}$.
        Again, it follows that $X^{k}_s(\minstrat^{k}) \in \uviop^k(\mathcal{X}^0)_s$.
    \end{itemize}

\subsection{Proof of Theorem~\ref{thm:alg-correct}}

\newcommand{\X}{\mathcal{X}}

\thmalgcorrect*

\begin{proof}
    Recall the operation $\mu \colon \pset{\pcurves}^S \to \pset{\pcurves}^S$. $\mu(\X)$ removes all non-inclusion-minimal elements from each $\X_s$, i.e., $\mu(\X)_s = \{X \in \X_s \mid X \text{ is } \subseteq\text{-minimal in } \X_s\}$.
    We first show the following auxiliary claim:
    
    \begin{claim*}
        For all $\X \in \pset{\pcurves}^S$ such that $\X(s)$ is finite it holds that
        \begin{equation*}
            \mu(\uviop(\X)) = \mu(\uviop(\mu(\X)))~.
        \end{equation*}
    \end{claim*}
    \begin{claimproof}
        We first show that
        \[
            \mu(\uviop(\X))_s \subseteq \uviop(\mu(\X))_s
        \]
        for all $s \in S$.
        Case distinction:
        \begin{itemize}
            \item $s \in \Smin$.
            By definition, $\uviop(\X)_s = \bigcup_{t \in \choices(s)} \X_t$.
            Let $X \in \uviop(\X)_s$ be a minimum, i.e., $X \in \mu(\uviop(\X))_s$.
            But then $X$ must also be a minimum in at least one of the $\X_t$, $t \in \choices(s)$.
            Thus $X \in \bigcup_{t \in \choices(s)} \mu(\X)_t = \uviop(\mu(\X))_s$.

            \item $s \in \Smax$.
            Recall that by definition, $\uviop(\X)_s = \{\conv(\bigcup_{t \in \choices(s)} X_t) \mid X \inpw \X \}$.
            Let $X \in \uviop(\X)_s$ be a minimum, i.e., $X \in \mu(\uviop(\X))_s$.
            Then $X = \conv(\bigcup_{t \in \choices(s)} X_t)$ for some $X_t \in \X_t$.
            We claim that for each $t$, we can assume that $X_t$ is a minimum in $\X_t$, i.e., $X_t \in \mu(\X)_t$.
            Indeed, suppose that for some $t$, $X_t$ is not a minimum.
            We can then exchange this $X_t$ for a minimum $X_t' \in \mu(\X)_t$ to construct an $X' \subseteq X$.
            This however implies that $X' = X$ because $X$ is a minimum.
            Thus again, $X \in \{\conv(\bigcup_{t \in \choices(s)} X_t) \mid X \inpw \mu(\X) \} = \uviop(\mu(\X))_s$.

            \item $s \in \Sprob$.
            Similar to the previous case.
        \end{itemize}
        To conclude the proof, notice that for all finite sets $\X, \mathcal{Y} \in \pset{\pcurves}$ with $\X \subseteq \mathcal{Y}$ and $\mu(\mathcal{Y}) \subseteq \mathcal{X}$ it already holds that $\mu(\X) = \mu(\mathcal{Y})$.
        It is clear that $\uviop(\mu(\X))_s \subseteq \uviop(\X)_s$ since $\mu(\X)_s \subseteq \X_s$.
        Putting all together yields $\mu(\uviop(\X))_s = \mu(\uviop(\mu(\X)))_s$ for all $s$ as claimed.
    
    \end{claimproof}
    We now complete the proof of Theorem~\ref{thm:alg-correct}.
    Let $\tilde{\uviop} \eqdef \mu \circ \uviop$.
    Technically, we show that
    \[
        \bigcap (\tilde{\uviop}^k(\mathcal{X}^0(\query))_s) = \bigcap (\uviop^k(\mathcal{X}^0(\query))_s)
    \]
    for all $s \in S$ and $k \geq 0$, i.e., iterating $\mu \circ \uviop$ instead of just $\uviop$ as in Lemma~\ref{lem:uvicorrect} is equivalent.
    To this end it suffices to show that
    \begin{equation}
        \label{eq:proofobl}
        \mu(\uviop^k(\mathcal{X}^0)) = \tilde{\uviop}^k(\mathcal{X}^0)
    \end{equation}
    because the intersection in completely determined by the $\subseteq$-minimal sets.
    We show \eqref{eq:proofobl} by induction:
    \begin{itemize}
        \item For $k = 0$, $\mu(\X^0) = \X^0$ clearly holds because $\X^0_s$ is a singleton for each $s \in S$.
        \item Let $k > 1$. Then
        \begin{align*}
            \mu(\uviop^k(\X^0)) &= \mu(\uviop(\uviop^{k-1}(\X^0)) \\
            & = \mu(\uviop(\mu(\uviop^{k-1}(\X^0))) \tag{Auxiliary claim above} \\
            & = \mu(\uviop(\tilde{\uviop}^{k-1}(\X^0)) \tag{induction hypothesis} \\
            & = \tilde{\uviop}^k(\X^0) \tag{definition}
        \end{align*}
        which completes the proof.
    \end{itemize}
\end{proof}

\subsection{Details on the Experimental Setup}
\label{app:exp-details}

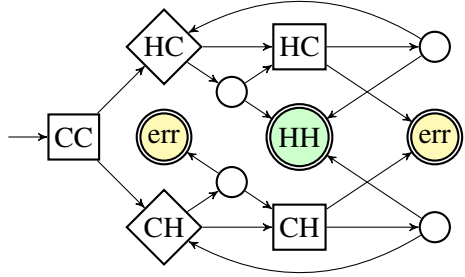
\begin{figure}
    \centering
    \begin{tikzpicture}[on grid, node distance = 12mm and 18mm,initial text=]
    \node[max,initial] (cc0) {CC};
    \node[prob,accepting,col1,right=12mm of cc0] (err0) {err};
    \node[min,above=of err0] (hc0) {HC};
    \node[min,below=of err0] (ch0) {CH};
    \node[max,right=of hc0] (hc1) {HC};
    \node[max,right=of ch0] (ch1) {CH};
    \node[prob,col2,accepting,right=of err0] (hh) {HH};
    \node[prob,right=of hc1,minimum size=4mm] (phc1) {};
    \node[prob,right=of ch1,minimum size=4mm] (pch1) {};
    \node[prob,accepting,col1,right=of hh] (err1) {err};
    \node[prob,right=9mm of hc0,yshift=-6mm,minimum size=4mm] (phc0) {};
    \node[prob,right=9mm of ch0,yshift=6mm,minimum size=4mm] (pch0) {};
    
    \draw[trans] (cc0) -- (hc0);
    \draw[trans] (cc0) -- (ch0);
    \draw[trans] (hc0) -- (phc0);
    \draw[trans] (hc0) -- (hc1);
    \draw[trans] (ch0) -- (pch0);
    \draw[trans] (ch0) -- (ch1);
    \draw[trans] (phc0) -- (hc1);
    \draw[trans] (phc0) -- (hh);
    \draw[trans] (pch0) -- (ch1);
    \draw[trans] (pch0) -- (err0);
    \draw[trans] (hc1) -- (phc1);
    \draw[trans] (hc1) -- (err1);
    \draw[trans] (ch1) -- (err1);
    \draw[trans] (ch1) -- (pch1);
    \draw[trans] (phc1) edge[bend right] (hc0);
    \draw[trans] (phc1) -- (hh);
    \draw[trans] (pch1) -- (hh);
    \draw[trans] (pch1) edge[bend left] (ch0);
    
    \end{tikzpicture}
    \caption{Adapted floor heating model from~\cite{BF16}.}
    \label{fig:floor-heating}
\end{figure}

The floor heating example from Figure~\ref{fig:floor-heating} is as in \cite{BF16} with the following modifications:
\begin{itemize}
    \item The controller can turn on at most one of the heatings, otherwise it will run out of energy and enter state $err$.
    \item The objective is the CQ (template) $\probability(\reach \{HH\}) \geq x \wedge \probability(\safe S \setminus \{err\}) \geq y$.
\end{itemize}
The random games together with a 2-dimensional mixed query each were constructed according to the following scheme:
\begin{itemize}
    \item Generate $m$ states.

    \item Assign each state to $\Smin, \Smax$ or $\Sprob$ with equal probability.

    \item Assign two random successors $s_1, s_2$ to each state $s$ such that $s_1,s_2,s$ are pairwise different.

    \item Select $l$ states at random, turn them into sinks and create a target set with those states.
    This will be the reachability objective.

    \item Similarly, create a safety objective in the same way by selecting $l$ random unsafe states.
\end{itemize}
The numbers in Table~\ref{tab:exp} are for $m = 10$ and $l = 1$.

}{}

\end{document}